\documentclass[12pt]{article}

\usepackage[utf8]{inputenc}
\usepackage[english]{babel}

\usepackage{titlesec}

\titleformat{\section}[block]
{\fontsize{12}{15}\bfseries\sffamily}
{\thesection}
{1em}
{}
\titleformat{\subsection}[hang]
{\fontsize{12}{15}\bfseries\sffamily}
{\thesubsection}
{1em}
{}

\usepackage{amsthm}
\usepackage{amsmath}

\newtheorem{theorem}{Theorem}[section]
\newtheorem{lemma}[theorem]{Lemma}
\newtheorem{corollary}[theorem]{Corollary}

\theoremstyle{definition}
\newtheorem{definition}[theorem]{Definition}

\theoremstyle{remark}

\usepackage{graphicx}
\usepackage{epstopdf}
\usepackage{caption}
\usepackage{subcaption}
\usepackage{float}

\renewcommand\footnotemark{}

\RequirePackage{bbding} 

\usepackage{csquotes}

\usepackage[numbers]{natbib}

\usepackage[colorlinks = true,
linkcolor = blue,
urlcolor  = blue,
citecolor = blue,
anchorcolor = blue]{hyperref}

\def \myURL {http://www19.homepage.villanova.edu/jesse.frey/}
\RequirePackage{fix-cm}
\makeatletter
\def\munderbar#1{\underline{\sbox\tw@{$#1$}\dp\tw@\z@\box\tw@}}
\makeatother

\begin{document}	
	
\title{\textbf{Design-based estimators of distribution function in ranked set sampling with an application}}
\author{ Yusuf Can Sevil\textsuperscript{a} and Tugba Ozkal Yildiz\textsuperscript{b}}

\date{\small \textsuperscript{a}The Graduate School of Natural and Applied Sciences, Dokuz Eylul University, Tinaztepe Campus, 35160, Buca, Izmir, Turkey \\
\textsuperscript{b}Department of Statistics, Faculty of Science, Dokuz Eylul University, Tinaztepe Campus, 35160, Buca, Izmir, Turkey}

\maketitle

\begin{abstract}
	Empirical distribution functions (EDFs) based on ranked set sampling (RSS) and its modifications have been examined by many authors. In these studies, the proposed estimators have been investigated for infinite population setting. However, developing EDF estimators in finite population setting would be more valuable for areas such as environmental, ecological, agricultural, biological, etc. This paper introduces new EDF estimators based on level-0, level-1 and level-2 sampling designs in RSS. Asymptotic properties of the new EDF estimators have been established. Numerical results have been obtained for the case when ranking is imperfect under different distribution functions. It has been observed that level-2 sampling design provides more efficient EDF estimator than its counterparts of level-0, level-1 and simple random sampling. In real data application, we consider a pointwise estimate of distribution function and estimation of the median of sheep's weights at seven months using RSS based on level-2 sampling design.
\end{abstract}

\vspace{0.2cm}
\hspace{-0.6cm}\textbf{Key Words:} Design-based estimators; Empirical distribution function; Finite population; H\'{a}jek type estimator; Ranked set sampling; Sheep data. \\
\hspace{-0.6cm}\textbf{Mathematics Subject Classifications 2020:} 62G30, 62D05, 65C05, 62P12

\vspace{2cm}
\section{Introduction}

Ranked set sampling (RSS) was proposed by McIntyre \cite{mcintyre1952method}. In this study, RSS was used for seeking to estimate the yield of pasture in Australia, effectively. Because, making precise yield measurements requires harvesting the crops and so it is expensive. McIntyre \cite{mcintyre1952method} proved that mean estimator of RSS is unbiased regardless of any error in ranking process. Then, Halls and Dell \cite{halls1966trial} examined the performance of RSS for estimating the weights of browse and of herbage in a pine. Also, they investigated the effect of ranking errors in practice. Theoretical properties of RSS were investigated by Takahasi and Wakimoto \cite{takahasi1968unbiased}. In this study, they showed that the mean estimator of RSS is unbiased and the variance of the estimator is always smaller than the variance of the mean estimator of simple random sampling (SRS) under perfect ranking. Dell and Clutter \cite{dell1972ranked} evaluated the effect of ranking errors on RSS. For the detailed literature, see Kaur et al. \cite{kaur1995ranked}, Chen et al. \cite{chen2003ranked}, Al-Omari and Bouza \cite{al2014review} and Bouza and Al-Omari \cite{bouza2018ranked}.

Let $X_{1},\cdots,X_{N}$ be population units with absolutely continuous distribution function $F(x)$,
\begin{equation}
	F\left( x\right) =\frac{1}{N}\sum\limits_{i=1}^{N}I\left( X_{i}\leq
	x\right).
\end{equation}
It is assumed that units are easy to rank and difficult and/or costly to measure. Also, $N$ is deemed to be fairly large. According to McIntyre \cite{mcintyre1952method}'s definition, the procedure of RSS as follows. A set of size $k$ is selected without replacement from the population. Then, these units are ranked from the smallest to the largest and the first smallest unit, say $X_{1(1)}$, is selected for full measurement. Here, the subscript $r(r)$ indicate that $X$ is the $r$th smallest unit in the $r$th set. Also, the parenthesis $()$ is used for the case when the ranking is perfect. If there is a doubt that the ranking is imperfect, the bracket $[]$ is used instead of the parenthesis $()$. After the first smallest unit is measured, the second set of size $k$ is selected without replacement. In the set, the second smallest unit, say $X_{2(2)}$, is selected for full measurement. This procedure is repeated until the largest unit, say $X_{k(k)}$, is selected from the $k$th set. To obtain a ranked set sample of size $n=mk$, the same steps are repeated $m$ times. The notation $m$ is called as the number of cycles. The units in the ranked set sample are denoted by $X_{r(r)j}$, $r=1,\cdots,k$ and $j=1,\cdots,m$. Thus, $m$ units are selected from each $r$th order statistic having distribution function $F_{(r)}(x)$, $r=1,\cdots,k$.

In the literature of RSS, estimation of distribution function is a remarkable topic. First, Stokes and Sager \cite{stokes1988characterization} suggested the following estimator,
\begin{equation}
	\hat{F}_{RSS}(x)=\frac{1}{mk}\sum\limits_{j=1}^{m}\sum\limits_{r=1}^{k}I%
	\left( X_{r(r)j}\leq x\right).
\end{equation}
They showed that $\hat{F}_{RSS}(x)$ is unbiased with variance
\begin{equation*}
	V\left[\hat{F}_{RSS}(x)\right]=\frac{1}{mk^{2}}\sum%
	\limits_{r=1}^{k}F_{(r)}(x)\left( 1-F_{(r)}(x)\right)
\end{equation*}
and
\begin{equation*}
	\frac{\hat{F}_{RSS}(x)-E\left( \hat{F}_{RSS}(x)\right) }{\left( V\left[\hat{F}_{RSS}(x)\right]\right)^{1/2}}
\end{equation*}
converges in distribution to standard normal as $m\rightarrow \infty$, when $x$ and $k$ are held fixed. Then, the estimator is modified by using other ranked based sampling methods such as extreme RSS \cite{samawi2001estimation}, double RSS \cite{abu2002distribution}, extreme median RSS \cite{Kim205edf}, L RSS \cite{al2015efficiency}, quartile RSS \cite{al2016quartile}, partially rank-ordered set \cite{nazari2016distribution}, pair RSS \cite{zamanzade2019edf} and percentile RSS \cite{sevil2021estimation}. 

In the literature, research in RSS draw considerable attention in finite population setting as well. Patil et al. \cite{patil1995finite} examined without replacement procedure in RSS to estimate the population mean. Deshpande et al. \cite{deshpande2006nonparametric} expanded the sampling design in Patil et al. \cite{patil1995finite} and introduced three different sampling designs which are level-0, level-1 and level-2 for the case when the population size $N$ is fairly small. These sampling designs are introduced in the Section 2. It is appeared that particular attention is paid to estimating first and second order inclusion probabilities in the literature since estimating the inclusion probabilities is a major problem to obtain some estimators such as Horvitz-Thompson mean and total estimators, Horvitz and
Thompson \cite{horvitz1952generalization}. Here, the first order inclusion probability $\pi_{i}$ is probability that $i$th population unit is appeared in the sample while the second order inclusion probability $\pi_{ii^{\prime}}$ is probability that both $i$th and $i^{\prime}$th population units are appeared in the sample, $1\leq i, i^{\prime}\leq N$. The first order inclusion probabilities are the same for all population units under SRS or under level-2 sampling design while the population units have different inclusion probabilities under level-0 sampling design or under level-1 sampling design. For the level-0 sampling design, Jafari Jozani and Johnson \cite{jafari2011design} obtained the first and second order inclusion probabilities. Also, they developed mean and ratio estimators based on level-0 sampling design. Inclusion probabilities for level-1 sampling design was investigated by Al-Saleh and Samawi \cite{al2007note} and Ozdemir and Gokpinar \cite{oezdemir2007generalized}. In these studies, authors estimated the inclusion probabilities by using conditional probabilities of all possible selections in the population. Therefore, the proposed estimators are useful only for small sample size. Then, Frey \cite{frey2011recursive} suggested a recursive algorithm that can used even if the sample size is relatively large. 

Recently, statistical inference for mean, total, variance and quantiles have been discussed by Ozturk \cite{ozturk2014statistical, ozturk2014estimation, ozturk2016estimation}, Ozturk and Bayramoglu Kavlak \cite{ozturk2018model, ozturk2021model} in finite population setting. Also, Sevil and Yildiz \cite{sevil2017power, sevil2020performances} and Yildiz and Sevil \cite{yildiz2018performances, yildiz2019empirical} proposed the following empirical distribution function (EDF) in finite population setting.
\begin{equation*}
	\hat{F}_{L-t}(x)=\frac{1}{mk}\sum\limits_{j=1}^{m}\sum\limits_{r=1}^{k}I%
	\left( X_{r(r)j}^{(t)}\leq x\right)
\end{equation*}
where $t=0$, $1$ and $2$ for level-0, level-1 and level-2, respectively. They proved that the EDF is unbiased and is more efficient than the EDF based on SRS. Moreover, the EDFs based on the sampling designs was applied to air quality data by Yildiz and Sevil \cite{yildiz2019empirical} and to body mass index by Sevil and Yildiz \cite{sevil2020performances}.

Ozturk et al. \cite{ozturk2005estimation} showed that the procedure of RSS can be used to collect data from farm animals, such as sheep, cattle and cows. In these experiments, taking measurements from any variables such as milk, meat and wool yields is time-consuming because of the size and physical behavior of the animals. Considering that easy to measure variables (including cheap and/or inexpensive observations) are often available, it can be said that RSS is a useful sampling technique in these experiments. By using the cheap and/or inexpensive observations, the interested variable is divided into several artificial strata without taking measurement. Thus, observations which are worthy to be measured are selected from each stratum. 

In the present paper, we investigate the distribution function of sheep weight (kg) from the Research Farm of Ataturk University, Erzurum, Turkey. Examining the quantiles can be provided valuable information for research. If the distribution function of the sheep weights is estimated, both the quantiles and the probabilities corresponding to the specific quantiles can be obtained. Unlike the studies Sevil and Yildiz \cite{sevil2017power, sevil2020performances} and Yildiz and Sevil \cite{yildiz2018performances, yildiz2019empirical}, we provide design-based estimators for distribution function. These estimators use the information of inclusion probabilities as well. Therefore, this paper brings a different perspective on estimation of distribution function. We give details about the new estimators in further sections. Rest of the paper organized as follows. Also, we define the procedures of the sampling designs in Section 2. Then, we examine the design-based estimators and theoretical properties of the estimators in Section 3. In Section 4, asymptotic properties of the proposed estimators are provided. Also, the asymptotic relative efficiencies of the proposed EDF estimators based on level-0, level-1 and level-2 sampling designs w.r.t the estimator based on SRS are presented in this section. In Section 5, empirical results for imperfect ranking are reported. In Section 6, the EDF estimator which has the best performance is applied to sheep data to estimate the distribution of their weights. Finally, we give some concluding remarks in Section 7.

\section{Sampling designs}

In this section, we describe the procedure of the sampling designs. These sampling designs have different procedures in terms of their replacement policies. Let $X_{1},\cdots,X_{N}$ be population units. The sampling design was defined by Deshpande et al. \cite{deshpande2006nonparametric} as follows:
\begin{itemize}
	\item Level-0 sampling design:
	\begin{enumerate}
		\item[(1)] Select $k$ units without replacement from the population.
		\item[(2)] By using a single auxiliary variable, rank the units and measured the $r$th order statistic from the $r$th set, $X^{(0)}_{r(r)j}$.
		\item[(3)] All units are replaced back into the population.
		\item[(4)] (1)-(3) are repeated $k$ times for $r=1,\cdots,k$.
		\item[(5)] (1)-(4) are repeated $m$ times for $j=1,\cdots,m$.
	\end{enumerate}
	\item Level-1 sampling design:
	\begin{enumerate}
		\item[(1)] Select $k$ units without replacement from the population.
		\item[(2)] By using a single auxiliary variable, rank the units and measured the $r$th order statistic from the $r$th set, $X^{(1)}_{r(r)j}$.
		\item[(3)] The measured unit is not replaced back into the population.
		\item[(4)] (1)-(3) are repeated $k$ times for $r=1,\cdots,k$.
		\item[(5)] (1)-(4) are repeated $m$ times for $j=1,\cdots,m$.
	\end{enumerate}
	\item Level-2 sampling design:
	\begin{enumerate}
		\item[(1)] Select $k$ units without replacement from the population.
		\item[(2)] By using a single auxiliary variable, rank the units and measured the $r$th order statistic from the $r$th set, $X^{(2)}_{r(r)j}$.
		\item[(3)] None of the units in the set are replaced back into the population.
		\item[(4)] (1)-(3) are repeated $k$ times for $r=1,\cdots,k$.
		\item[(5)] (1)-(4) are repeated $m$ times for $j=1,\cdots,m$.
	\end{enumerate}
\end{itemize}

For the level-0 sampling design, a population unit may be selected more than once both in the ranking process and in the final sample. In the level-1 sampling design, a population unit may be appeared in the ranking process, but is not be appeared in the final sample. If the level-2 sampling design is used, a unit in the population appeared more than once neither in the ranking process nor in the final sample. Also, we need $k\leq N$ for the level-0 sampling design, $\max_\iota\left\lbrace \left(\iota-1 \right) +k \right\rbrace \leq N$ for level-1 sampling design, and $mk^{2}\leq N$ for the level-2 sampling design, Frey \cite{frey2011recursive} where $\iota=1,\cdots,mk$.

The sampling designs have similar behaviors for large population ($N$ is fairly large), but they perform differently for small finite population.  Deshpande et al. \cite{deshpande2006nonparametric} proposed nonparametric confidence intervals based on these sampling designs for population median. They recommended the level-2 sampling design for small finite populations such as $N=20$, $30$, $40$ and $50$.

\section{Design-based estimators}

Assumed that $X_{1},\cdots,X_{N}$ is population units. We supposed that $\pmb{D}$ is a sample of size $n$ which is selected from the population by using a sampling design with $\pi_{i}$($>0$) and $\pi_{ii^{\prime}}$($>0$), $1\leq i,i^{\prime}\leq N$. For estimation of $F(x)$ that is given by Eq. (1), the following estimator has been proposed.
\begin{equation}
	\hat{F}_{h}(x)=\sum\limits_{i\in \pmb{D}}\frac{\frac{I\left( X_{i}\leq x\right) 
		}{\pi _{i}}}{\sum\limits_{\imath\in \pmb{D}}\frac{1}{\pi _{\imath}}},
\end{equation}
where $1\leq i\leq N$. The estimator is called H\'{a}jek type estimator, Arnab \cite{arnab2017survey}. The first and second order inclusion probabilities, $\pi _{i}$ and $\pi _{ii^{\prime}}$ vary depending on the sampling design. In the present paper, the sample units in $\pmb{D}$ is obtained by using SRS without replacement, so the first and second order inclusion probabilities are $\pi _{i}= n/N$ and $\pi _{ii^{\prime}}= n(n-1)/(N(N-1))$ for $1\leq i,i^{\prime}\leq N$, respectively. The properties of H\'{a}jek type estimator are as follows:
\begin{theorem}
	Let $\pmb{D}$ is a sample of size $n$ and is obtained by using SRS without replacement design. Then, 
	\begin{enumerate}
		\item[1.] $\hat{F}_{h}(x)$ is unbiased for $F(x)$.
		\item[2.] 
		\begin{equation}
			\begin{split}
				V\left[ \hat{F}_{h}(x)\right] =N^{-2}\sum\limits_{i=1}^{N}\sum%
				\limits_{i^{\prime }=1}^{N} & \left( \pi _{ii^{\prime }}-\pi _{i}\pi
				_{i^{\prime }}\right) \\ 
				&\hspace{-1.5cm}\times \left( \frac{I\left( X_{i}\leq x\right) -F\left(
					x\right) }{\pi _{i}}\right) \left( \frac{I\left( X_{j}\leq x\right) -F\left(
					x\right) }{\pi _{i^{\prime }}}\right)
			\end{split}
		\end{equation}  
		\item[3.] 
		\begin{equation}
			\begin{split}
				\hat{V}\left[ \hat{F}_{h}(x)\right] =-\frac{1}{2\left( \sum\limits_{i\in D}\frac{1}{\pi _{i}}\right)^{-2}}%
				\sum\limits_{i\in D}\sum\limits_{i^{\prime }\in D} & \frac{\left( \pi
					_{ii^{\prime }}-\pi _{i}\pi _{i^{\prime }}\right) }{\pi _{ii^{\prime }}} \\ 
				&\hspace{-3cm}\times
				\left( \frac{I\left( X_{i}\leq x\right) -\hat{F}_{h}(x)}{\pi _{i}}-\frac{%
					I\left( X_{j}\leq x\right) -\hat{F}_{h}(x)}{\pi _{i^{\prime }}}\right)^{2}
			\end{split}
		\end{equation} 
	\end{enumerate}
	where $1\leq i,i^{\prime}\leq N$. 
\end{theorem}

\begin{proof}
	\begin{enumerate}
		\item[1.] Recall that $\pi _{i}= \frac{n}{N}$ for SRS without replacement,
		\begin{equation}
			\begin{split}
				E\left[ \hat{F}_{h}(x)\right] =&\sum\limits_{i\in D}\frac{\frac{E\left[
						I\left( X_{i}\leq x\right) \right] }{n/N}}{\sum\limits_{\imath\in D}\frac{1}{n/N}%
				} \\ 
				=&\frac{1}{N}\sum\limits_{i\in D}\frac{N}{n}F(x) \\
				=& F(x)
			\end{split}
		\end{equation}
		Item $2$ and $3$ can be found in S{\"a}rndal et al. \cite{sarndal2003model} (formula (5.11.7) p.202) and (formula (5.11.9) p.203), respectively.
	\end{enumerate}
\end{proof}

Without loss of generality, it is supposed that $X_{1}\leq \cdots \leq X_{N}$. Here, ranking the population units can be performed by using an available auxiliary variable such as outcomes of previous experiment or another study variable which is correlated with interested variable $X$. Measurements of the auxiliary variable is deemed to require no additional cost or be easier and/or much cheaper than the interested variable. On the other hand, the ranking process can be done via visual inspection of the units if an auxiliary variable is not available. In both cases, we only need consistent ranking scheme to determine ranks of the population units. In other words, there must be a stable ranking of the units so that such ranks are well defined. Considering that the obtained ranked set sample based on level-t sampling design is denoted by $\pmb{D}_{t}$ which includes $n=mk$ units, the proposed estimators are expressed as follows:
\begin{equation}
	\hat{F}_{L-t}^{\ast }(x)=\sum\limits_{i\in \pmb{D}_{t}}\frac{\frac{
			I\left( X_{i}\leq x\right) }{\pi^{(t)}_{i}}}{\sum\limits_{\imath\in \pmb{D}_{t}}\frac{1}{%
			\pi^{(t)}_{\imath}}}
\end{equation}
where $1\leq i\leq N$ and $t=0,1,2$. The properties of the proposed estimators are as follows:
\begin{theorem}
	Let $\pmb{D}_{t}$ is a ranked set sample and is obtained by using level-t sampling design for $t=0,1,2$. Under a consistent ranking scheme, 
	\begin{enumerate}
		\item[1.] $\hat{F}_{L-0}^{\ast }(x)$ and $\hat{F}_{L-2}^{\ast }(x)$ are unbiased but $\hat{F}_{L-1}^{\ast }(x)$ is approximately unbiased for $F(x)$.
		\item[2.] 
		\begin{equation}
			\begin{split}
				V\left[ \hat{F}_{L-t}^{\ast }(x)\right] =N^{-2}\sum\limits_{i=1}^{N}\sum%
				\limits_{i^{\prime }=1}^{N} & \left( \pi^{(t)} _{ii^{\prime }}-\pi^{(t)} _{i}\pi^{(t)}
				_{i^{\prime }}\right) \\ 
				&\hspace{-1.5cm}\times \left( \frac{I\left( X_{i}\leq x\right) -F\left(
					x\right) }{\pi^{(t)} _{i}}\right) \left( \frac{I\left( X_{i^{\prime }}\leq x\right) -F\left(
					x\right) }{\pi^{(t)} _{i^{\prime }}}\right)
			\end{split}
		\end{equation}
		where $1\leq i,i^{\prime}\leq N$.  
		\item[3.] 
		\begin{equation}
			\begin{split}
				\hat{V}\left[ \hat{F}_{L-t}^{\ast }(x)\right] =-\frac{1}{2\left( \sum\limits_{i\in \pmb{D}_{t}}\frac{1}{\pi^{(t)} _{i}}\right)^{-2}}%
				\sum\limits_{i\in \pmb{D}_{t}}\sum\limits_{i^{\prime }\in \pmb{D}_{t}} & \frac{\left( \pi^{(t)}
					_{ii^{\prime }}-\pi^{(t)} _{i}\pi^{(t)} _{i^{\prime }}\right) }{\pi^{(t)} _{ii^{\prime }}} \\ 
				&\hspace{-4cm}\times
				\left( \frac{I\left( X_{i}\leq x\right) -\hat{F}_{L-t}^{\ast }(x)}{\pi^{(t)} _{i}}-\frac{%
					I\left( X_{i^{\prime }}\leq x\right) -\hat{F}_{L-t}^{\ast }(x)}{\pi^{(t)} _{i^{\prime }}}\right)^{2}
			\end{split}
		\end{equation}
		\item[4.] $V\left[ \hat{F}_{L-t}^{\ast }(x)\right]\leq V\left[ \hat{F}_{h}(x)\right]$ for both perfect and imperfect ranking.
	\end{enumerate} 
\end{theorem}
\begin{proof}
	For the proposed estimators based on RSS designs to be unbiased, the sum of $\pi _{\imath}^{(t)}$s must be $N$. However, the first order inclusion probabilities differ depending on which unit is sampled especially for level-1 sampling design. Therefore, it is difficult to show the unbiasedness for the EDF based on level-1 sampling design. For level-0 and level-2 sampling designs $N$, $\sum\limits_{\imath\in D_{t}}1/\pi _{\imath}^{(t)}= N$,
	\begin{equation*}
		\begin{split}
			E\left[ \hat{F}_{L-t}^{\ast }(x)\right] =&\sum\limits_{i\in D_{t}}\frac{%
				\frac{E\left( I\left( X_{i}\leq x\right) \right) }{\pi _{i}^{(t)}}}{%
				N} \\ 
			=&\frac{F(x)}{N}\sum\limits_{i\in D_{t}}\frac{1}{\pi _{i}^{(t)}} \\
			=& F(x)
		\end{split}
	\end{equation*}
	where $t=0$ and $2$. On the other hand, we can say that the EDF based on level-1 sampling design is not unbiased but only approximately unbiased since $\sum\limits_{\imath\in D_{1}}1/\pi _{\imath}^{(1)}\approx N$ when ranking is performed consistently. The Item 2 and Item 3 hold by the usual theory in Section 5.11 of S{\"a}rndal et al. \cite{sarndal2003model}. A proof is provided for the Item 4 in appendix.
\end{proof}

To compute the Eqs. (7), (8) and (9), the first and the second order inclusion probabilities ($\pi^{(t)}_{i}$, $\pi^{(t)} _{ii^{\prime }}$) are calculated for $1\leq i,i^{\prime}\leq N$. Let $I\left( i;r,k,N\right)$ be the probability that the $ith$ ranked unit in a population of size $N$ has rank $r$ in a set of size $k$. 
\begin{equation*}
	I\left( i;r,k,N\right) =\frac{\binom{i-1}{r-1}\binom{N-i}{k-r}}{\binom{N}{k}%
	}
\end{equation*}
Jafari Jozani and Johnson \cite{jafari2011design} used the following lemma to estimate the inclusion probabilities for level-0 sampling design.
\begin{lemma}
	When the ranked set sample is obtained by using level-0 sampling design, the first order inclusion probability is
	\begin{equation*}
		\pi _{i}^{(0)}=1-\prod\limits_{h=1}^{n}\left( 1-I\left( i;r_{h},k,N\right)
		\right)
	\end{equation*}
	and the second order inclusion probability is
	\begin{equation*}
		\begin{split}
			\pi _{ii^{\prime }}^{(0)} =&1-\prod\limits_{h=1}^{n}\left( 1-I\left(
			i;r_{h},k,N\right) \right) -\prod\limits_{h=1}^{n}\left( 1-I\left(
			i^{\prime };r_{h},k,N\right) \right) \\ 
			& +\prod\limits_{h=1}^{n}\left(
			1-I\left( i;r_{h},k,N\right) -I\left( i^{\prime };r_{h},k,N\right) \right)
		\end{split}
	\end{equation*}
	where $r_{h}$ is the rank of the measured unit in the $h$th set. 
\end{lemma}
\hspace{-0.35cm}Estimation of the first and second order inclusion probabilities for level-1 are more complex than estimation of their counterparts for level-0 and level-2 sampling designs. Frey \cite{frey2011recursive} proposed recursive algorithm to compute the first and second order inclusion probabilities. This algorithm basically calculates conditional probabilities for each unit in the population. Then, the first and second order inclusion probabilities are computed using the conditional probabilities. For the detailed mathematical backgrounds of the recursive algorithm, see Frey \cite{frey2011recursive}. Also, R codes for the first and second order inclusion probabilities are available on the Frey's web site at \href{\myURL}{\nolinkurl{\myURL}}. On the other hand, the inclusion probabilities for level-2 sampling design can be obtained by using the following lemma. The lemma includes the results in Theorem 3.1 and Theorem 3.2 of Patil et al. \cite{patil1995finite}.
\begin{lemma}
	When the ranked set sample is obtained by using level-2 sampling design, the first order inclusion probability is
	\begin{equation*}
		\pi _{i}^{(2)}=\sum\limits_{h=1}^{n}\left( 1-I\left( i;r_{h},k,N\right)
		\right)
	\end{equation*}
	and, for $i<i^{\prime }$, the second order inclusion probability is
	\begin{equation*}
		\pi _{ii^{\prime }}^{(2)}=\sum\limits_{h=1}^{n}\sum\limits_{h\neq
			h^{\prime }}\sum\limits_{\lambda =0}^{h-h^{\prime }-1}\frac{\binom{i-1}{%
				r_{h}-1}\binom{i-i^{\prime }-1}{\lambda }\binom{N-i^{\prime }}{k-\lambda
				-r_{h}}\binom{i^{\prime }-1-r_{h}-\lambda }{r_{h^{\prime }}-1}\binom{%
				N-i^{\prime }-k+\lambda +r_{h}}{k-r_{h^{\prime }}}}{\binom{N}{k}\binom{N-k}{k%
		}}
	\end{equation*}
	where $r_{h}$ and $r_{h^{\prime }}$ are the ranks of the measured units in the $h$th and $h^{\prime }$th sets, respectively. 
\end{lemma}

\section{Large sample properties}

In this section, we investigate the asymptotic properties of the proposed estimators under the case when the ranking is perfect. In the next section, we give the experimental results under the case when the ranking is imperfect. 

First, we recall the concepts of consistency and asymptotic normality from the general theory of statistical inference. Let $X_{1},\cdots,X_{n}$ be a simple random sample that is selected from a population of size $N$. Assumed that $\hat{\theta}_{n}$ is an unbiased estimator of $\theta$ based on SRS. When the $N$ is large enough, the asymptotic properties of the estimator $\hat{\theta}_{n}$ are discussed as $n\rightarrow \infty$. According to central limit theorem, $\hat{\theta}_{n}$ is said to have an asymptotic normal distribution with mean $\theta$ and asymptotic variance $V\left[ \hat{\theta}_{n}\right]=\sigma^{2}/n$,
\begin{equation*}
	\frac{\hat{\theta}_{n}-\theta }{\sigma /\sqrt{n}}\sim N\left(
	0,1\right)
\end{equation*}
as $n\rightarrow \infty$. Also, the estimator $\hat{\theta}_{n}$ is said to be consistent estimator of $\theta$ if for every $\epsilon>0$,
\begin{equation*}
	\underset{n\rightarrow \infty }{\lim} \Pr \left\{ \left\vert \hat{\theta}%
	_{n}-\theta \right\vert <\epsilon \right\} =1.
\end{equation*}

Now, let us consider that $N$ is small enough. In this case, it is supposed that both $n$ and $N$ tend to infinity for showing asymptotic results. Therefore, we first define a sequence of populations $U_{1},U_{2},U_{3},\cdots$ where $U_{v}$ includes $N_{v}$ units, $v=1,2,3,\cdots$. Consider that $U_{1}\subset U_{2}\subset U_{3}\cdots$ and so $N_{1}\leq N_{2}\leq N_{3}\leq \cdots$. Also, the distribution function of each population is denoted by $F_{v}(x)$,
\begin{equation*}
	F_{v}\left( x\right) =\frac{1}{N_{v}}\sum\limits_{i=1}^{N_{v}}I\left( X_{i}\leq
	x\right).
\end{equation*}
From each population, a ranked set sample of size $n_{v}=m_{v}k$ is selected by using level-t sampling design where $t=0,1,2$. Note that sample of sizes $n_{1}\leq n_{2}\leq n_{3}\leq \cdots$ increase since the numbers of cycles $m_{1}\leq m_{2}\leq m_{3}\leq \cdots$ are increased and the set size $k$ is fixed. We note that the set sizes $k_{1}\leq k_{2}\leq k_{3}\leq \cdots$ can also be increased to increase the size of the sample. Each ranked set sample is notated by $\pmb{D}_{t_{v}}$, $v=1,2,3,\cdots$ and $t=0,1,2$. Also, the inclusion probabilities, $\pi^{(t)}_{vi}$ and $\pi^{(t)}_{vii^{\prime}}$, are determined by the level-t sampling design where $1\leq i,i^{\prime}\leq N_{v}$. Thus, a sequence of estimators is defined as follows:
\begin{equation*}
	\hat{F}_{L-t_{v}}^{\ast }(x)=\sum\limits_{i\in \pmb{D}_{t_{v}}}\frac{\frac{
			I\left( X_{i}\leq x\right) }{\pi^{(t)}_{vi}}}{\sum\limits_{\imath\in \pmb{D}_{t_{v}}}\frac{1}{%
			\pi^{(t)}_{v\imath}}}
\end{equation*}
where $v=1,2,3,\cdots$ and $t=0,1,2$. The sequence of the variances of the estimators is  
\begin{equation*}
	\begin{split}
		V\left[ \hat{F}_{L-t_{v}}^{\ast }(x)\right] =N_{v}^{-2}\sum\limits_{i=1}^{N_{v}}\sum%
		\limits_{i^{\prime }=1}^{N_{v}} & \left( \pi^{(t)} _{vii^{\prime }}-\pi^{(t)} _{vi}\pi^{(t)}
		_{vi^{\prime }}\right) \\ 
		&\hspace{-1.5cm}\times \left( \frac{I\left( X_{i}\leq x\right) -F\left(
			x\right) }{\pi^{(t)} _{vi}}\right) \left( \frac{I\left( X_{j}\leq x\right) -F\left(
			x\right) }{\pi^{(t)} _{vi^{\prime }}}\right)
	\end{split}
\end{equation*}
where $1\leq i,i^{\prime}\leq N_{v}$.
Under the conditions, the Definition 4.1 determines that the proposed estimator $\hat{F}_{L-t}^{\ast }(x)$ is consistent estimator of $F(x)$.
\begin{definition}
	Let us define a sequence of estimators $\hat{F}_{L-t_{v}}^{\ast }(x)$ that are obtained from the sample $\pmb{D}_{t_{v}}$ where $t=0,1,2$ and $v=1,2,3\cdots$. Also, we know that $\pmb{D}_{t_{v}}$ is selected from the population $U_{v}$ of size $N_{v}$. Then, the estimator $\hat{F}_{L-t}^{\ast }(x)$ is said to be consistent estimator of $F(x)$ if for every $\epsilon>0$,
	\begin{equation*}
		\underset{v\rightarrow \infty }{\lim} \Pr \left\{ \left\vert \hat{F}_{L-t_{v}}^{\ast }(x)-F_{v}(x) \right\vert <\epsilon \right\} =1.
	\end{equation*}
\end{definition}

Some regularity conditions are given by Theorem 4.2 for Corollary 4.3. These conditions are results of the usual theory which is given by S{\"a}rndal et al. \cite{sarndal2003model} on p. 56.
\begin{theorem}
	While $v\rightarrow \infty$,
	\begin{enumerate}
		\item[1.] 
		\begin{equation}
			\frac{\hat{F}_{L-t_{v}}^{\ast }(x)-F_{v}(x)}{\left( V\left[ \hat{F}%
				_{L-t_{v}}^{\ast }(x)\right] \right) ^{1/2}}\rightarrow N(0,1).
		\end{equation}
		\item[2.] There exists a consistent variance estimator $\hat{V}\left[ \hat{F}_{L-t}^{\ast }(x)\right]$ for $V\left[ \hat{F}_{L-t}^{\ast }(x)\right]$.
	\end{enumerate}
\end{theorem}
\hspace{-0.35cm}The Corollary 4.3 follows the conditions in Theorem 4.2.
\begin{corollary}
	Let $\hat{F}_{L-t}^{\ast }(x)$ be a point in the sequence of estimators. Then, an approximate $100(1-\alpha)\%$ confidence interval for $F(x)$ can be defined as follows:
	\begin{equation}
		\hat{F}_{L-t}^{\ast }(x)\pm z_{\alpha/2}\left( \hat{V}\left[ \hat{F}%
		_{L-t}^{\ast }(x)\right] \right) ^{1/2}
	\end{equation}
	where $z_{\alpha/2}$ is the upper $100\left(\alpha/2\right)\%$ quantile of the $N(0,1)$ and $\hat{V}\left[ \hat{F}%
	_{L-t}^{\ast }(x)\right]$ is given by Eq. (9).
\end{corollary}
\begin{proof}
	While $v\rightarrow \infty$,
	\begin{equation*}
		\frac{\hat{F}_{L-t_{v}}^{\ast }(x)-F_{v}(x)}{\left( \hat{V}\left[ \hat{F}_{L-t_{v}}^{\ast }(x)\right] \right)^{1/2} } = \frac{\hat{F}_{L-t_{v}}^{\ast }(x)-F_{v}(x)}{\left( V\left[ \hat{F}_{L-t_{v}}^{\ast }(x)\right] \right)^{1/2} } \times \left( \frac{V\left[ \hat{F}_{L-t_{v}}^{\ast }(x)\right]}{\hat{V}\left[ \hat{F}_{L-t_{v}}^{\ast }(x)\right]}\right)^{1/2} 
	\end{equation*}
	Under the first condition of the Theorem 4.2, the first term of the left-hand side of the equation has an approximate $N(0,1)$. Under the second condition of Theorem 4.2, the second term of the left-hand side of the equation is
	\begin{equation*}
		\left( \frac{V\left[ \hat{F}_{L-t_{v}}^{\ast }(x)\right]}{\hat{V}\left[ \hat{F}_{L-t_{v}}^{\ast }(x)\right]}\right)^{1/2} \rightarrow 1. 
	\end{equation*}
	and this completes the proof.
\end{proof}

Now, we report some numerical results which are obtained for perfect ranking case. We investigate REs of EDFs based on level-t sampling designs to its counterpart in SRS by using the following equation.
\begin{equation*}
	RE(\hat{F}_{h}(x_{p}),\hat{F}_{L-t}^{\ast }(x_{p}))=\frac{V\left[ \hat{F}%
		_{h}(x_{p})\right] }{V\left[ \hat{F}_{L-t}^{\ast }(x_{p})\right] }
\end{equation*}
for $t=0,1,2$ where $x_{p}$ is the value of quantile corresponding to $F(x)=p\in\left\lbrace 0,0.1,0.2,\cdots,0.9,1 \right\rbrace $. We have considered three different populations which have $N=20$, $N=50$ and $N=100$ units, respectively. The set sizes are $k\in\left\lbrace 2,3,4 \right\rbrace$ for $N=20$ and $k\in\left\lbrace 2,3,4,5,6,7 \right\rbrace$ for $N=50$. For $N=100$, the number of cycles and the set sizes have been taken as $m\in\left\lbrace 2,4 \right\rbrace$ and $k\in\left\lbrace 3,5 \right\rbrace$, respectively. The values of REs are given by Figures \ref{fig1}-\ref{fig3}. According to the figures, it is appeared that the EDFs based on the sampling designs are more efficient than the EDF based on SRS. It is observed that the values of RE is monotone increasing as $F(x)$ increase for $0\leq F(x)\leq 0.5$. The EDF estimator based on level-2 sampling design has highest efficiencies among the all studied sampling methods except for $F(x)\geq 0.8$ when $N=20$ and $N=50$. It can be seen that the EDF based on level-1 is slightly more efficient than the EDF based on level-2 for $F(x)\geq 0.8$ when $N=20$ and $N=50$. 
\begin{figure}[http!]
	\centering
	\subfloat[For $k=2$.]{\includegraphics[width=6.8cm,height=7.5cm]{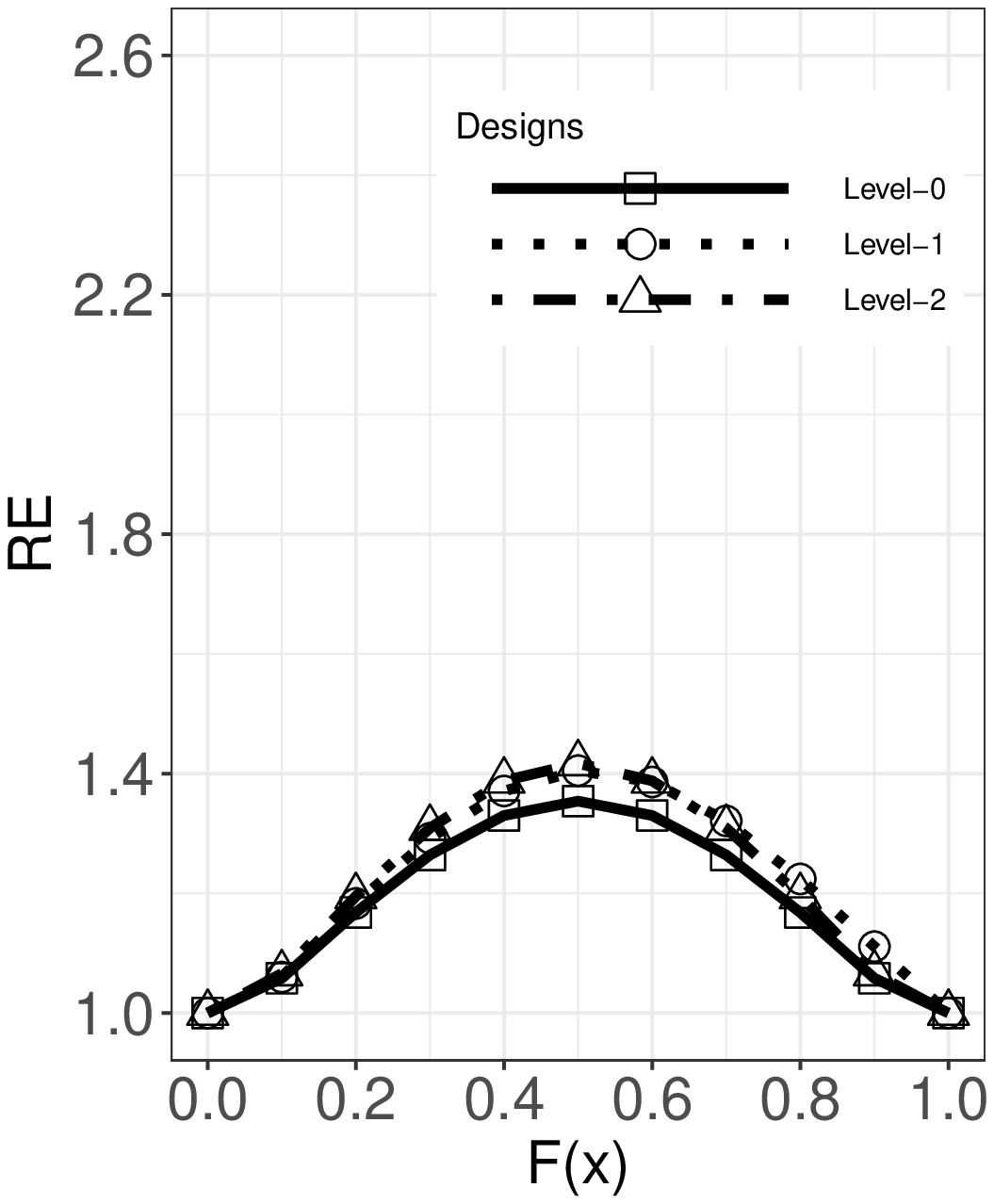}}
	\subfloat[For $k=3$.]{\includegraphics[width=6.8cm,height=7.5cm]{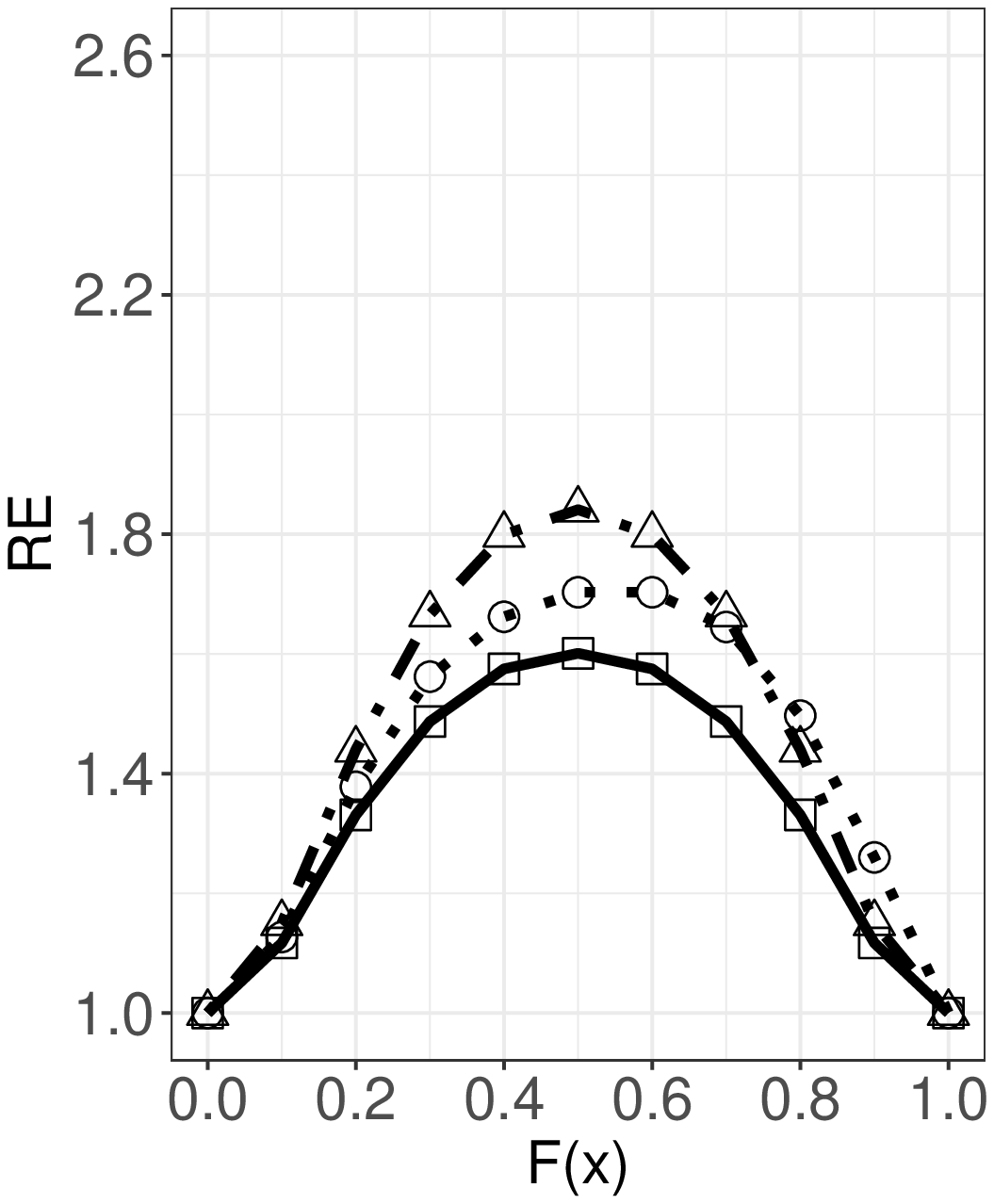}}
	\hskip2em
	\subfloat[For $k=4$.]{\includegraphics[width=6.8cm,height=7.5cm]{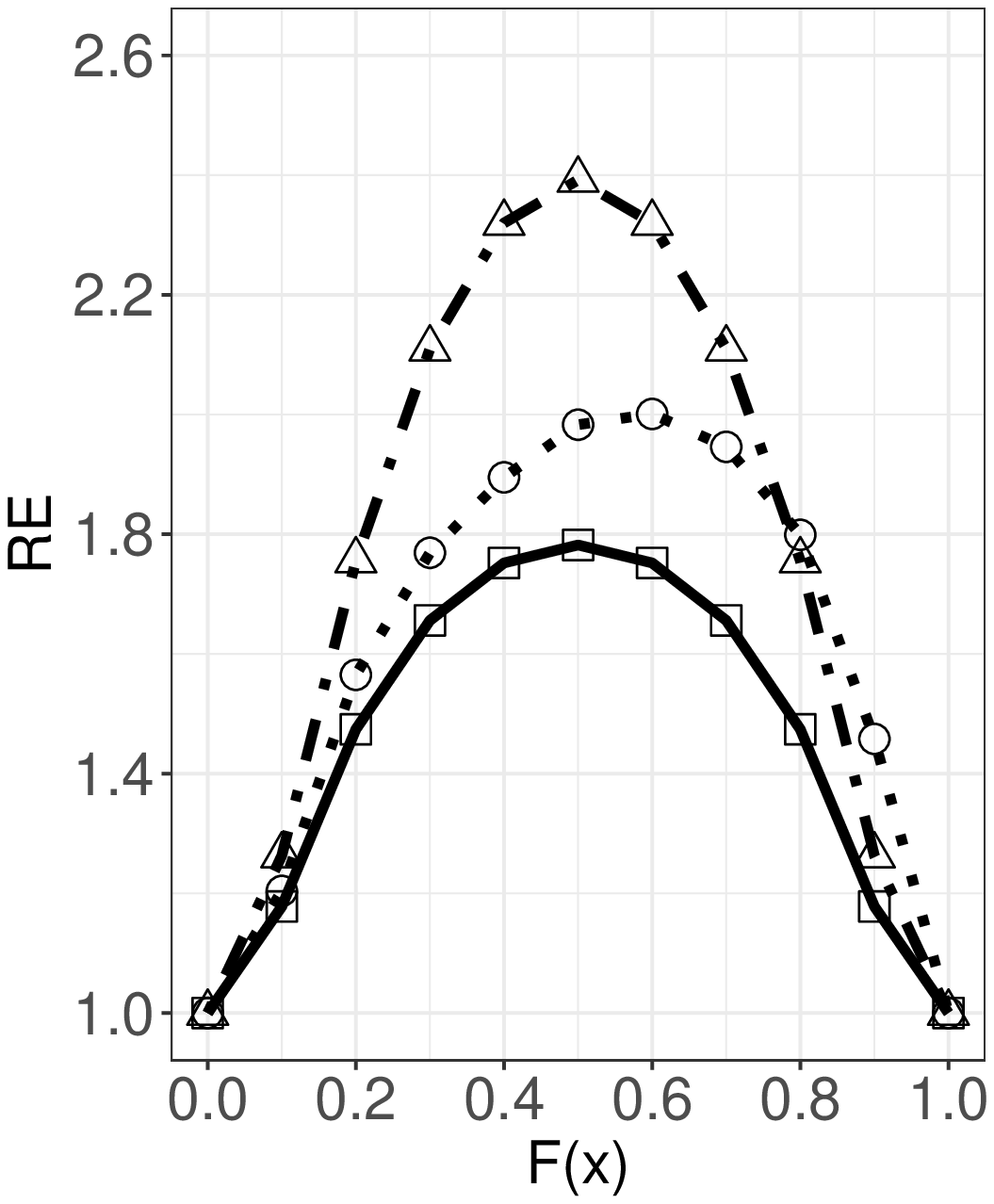}}
	\caption{For $N=20$, the values of $RE(\hat{F}_{h}(x_{p}),\hat{F}_{L-t}^{\ast }(x_{p}))$ for $t=0,1,2$ where $F(x)=p$}
	\label{fig1}
\end{figure}
\begin{figure}[http!]
	\centering
	\subfloat[For $k=2$.]{\includegraphics[width=6.8cm,height=5.6cm]{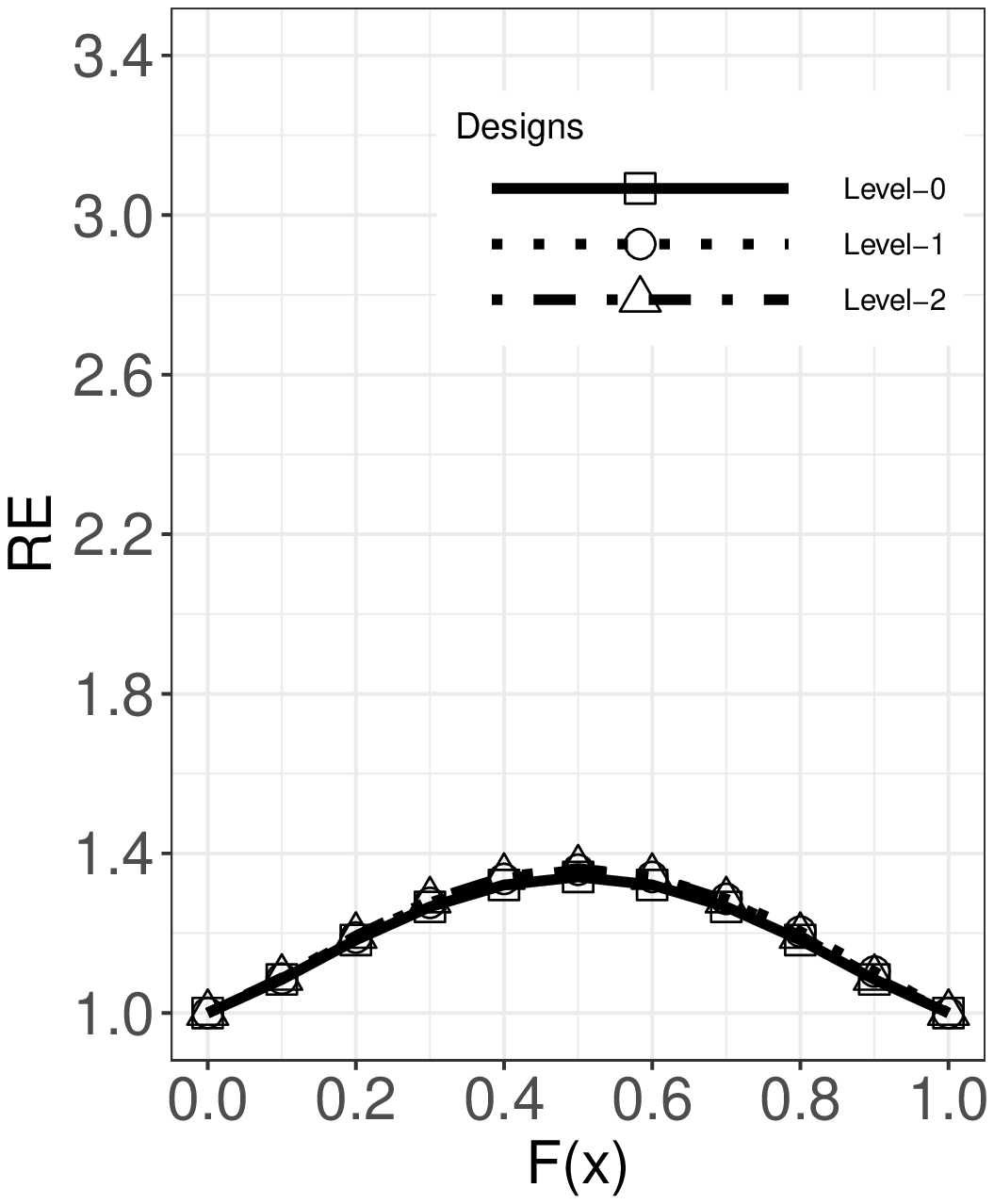}}
	\subfloat[For $k=3$.]{\includegraphics[width=6.8cm,height=5.6cm]{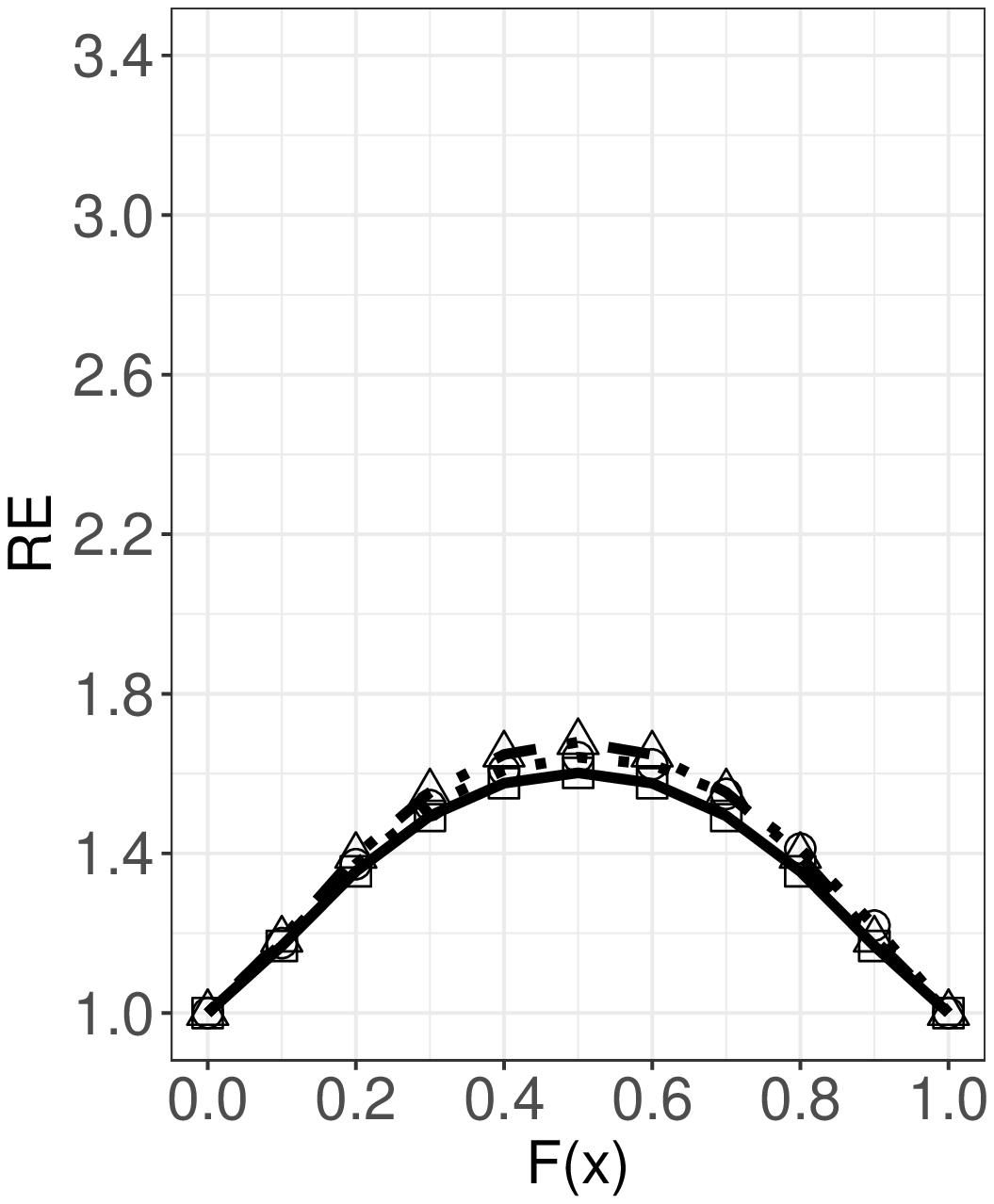}}
	\hskip2em
	\subfloat[For $k=4$.]{\includegraphics[width=6.8cm,height=5.6cm]{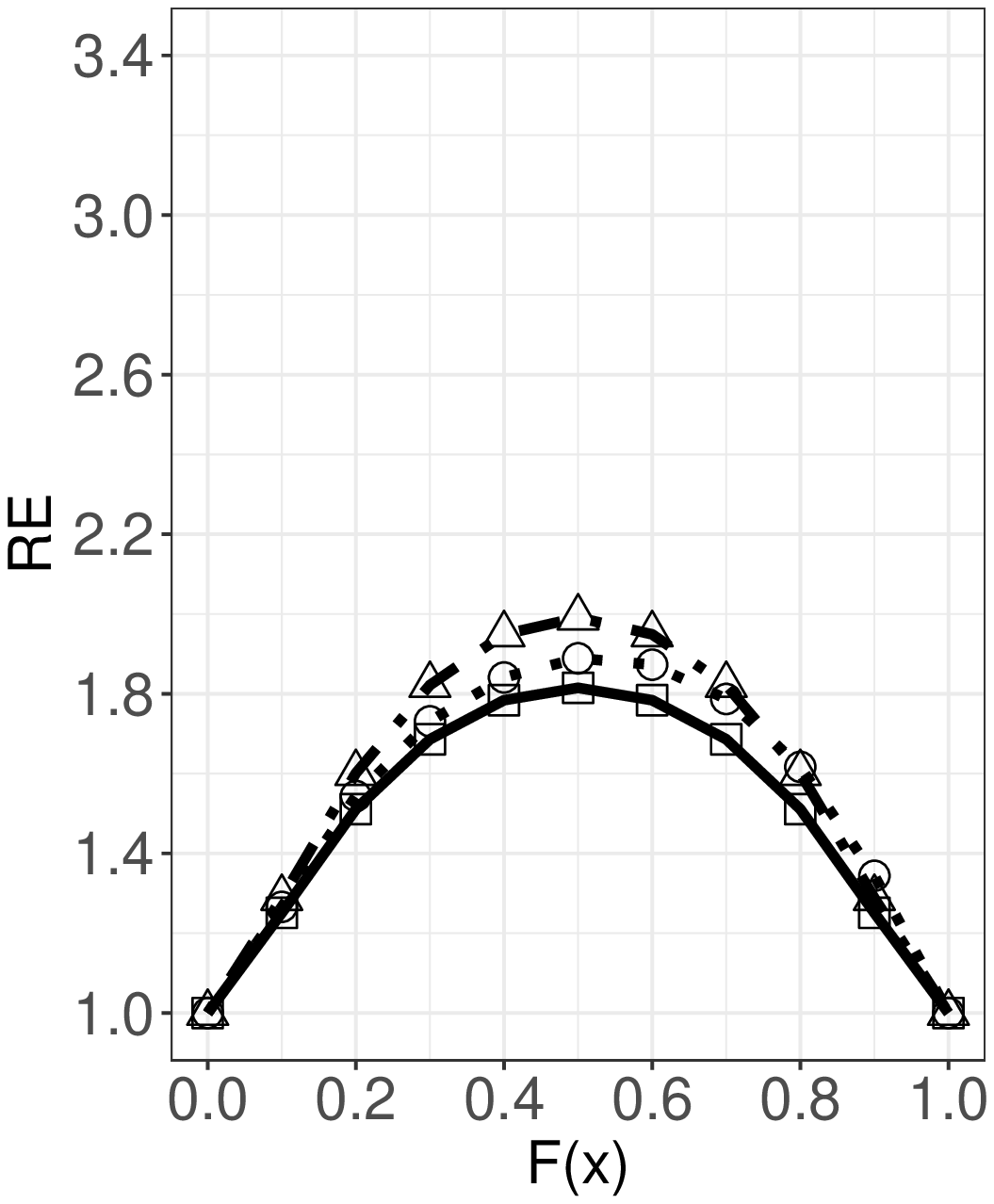}}
	\subfloat[For $k=5$.]{\includegraphics[width=6.8cm,height=5.6cm]{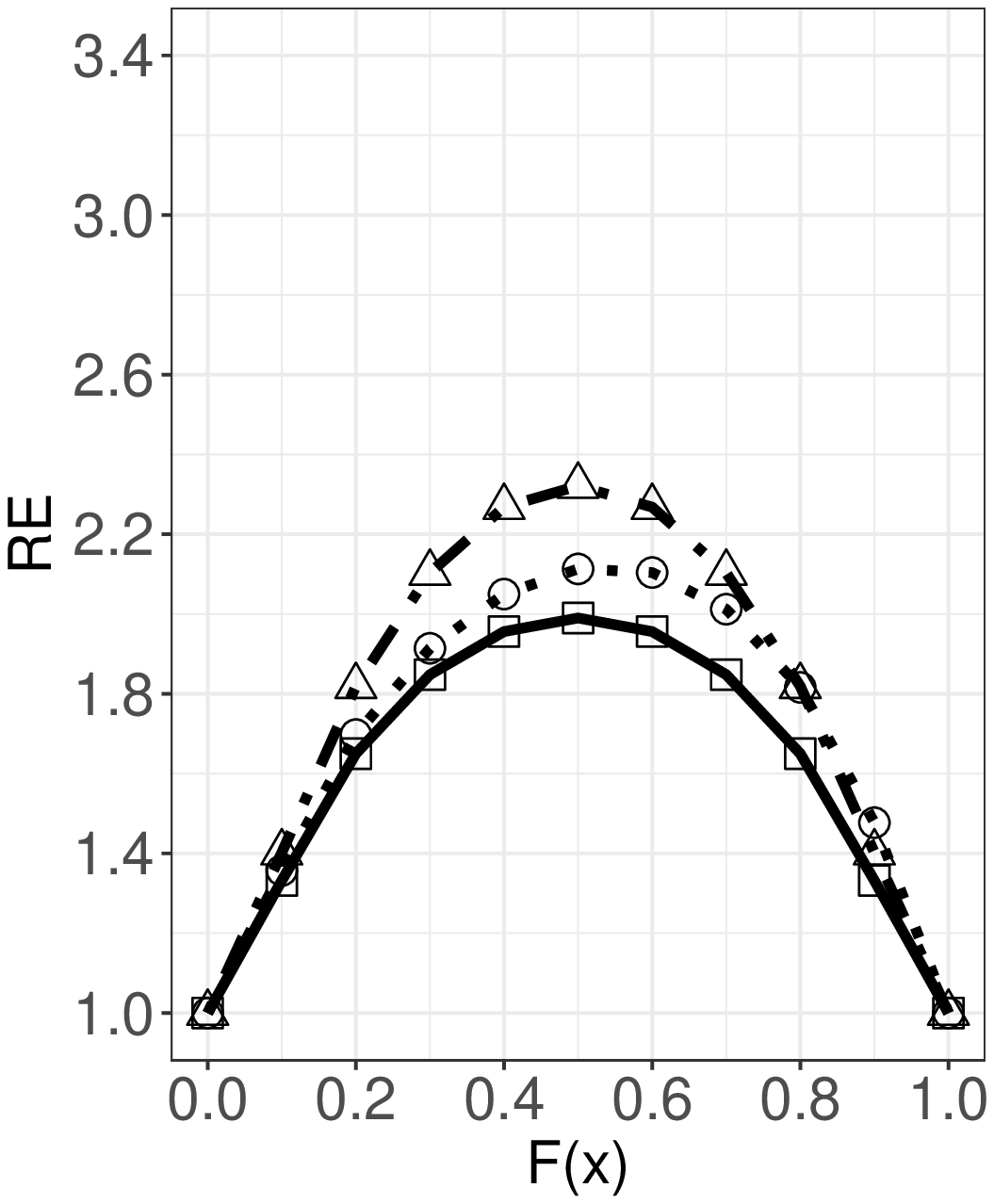}}
	\hskip2em
	\subfloat[For $k=6$.]{\includegraphics[width=6.8cm,height=5.6cm]{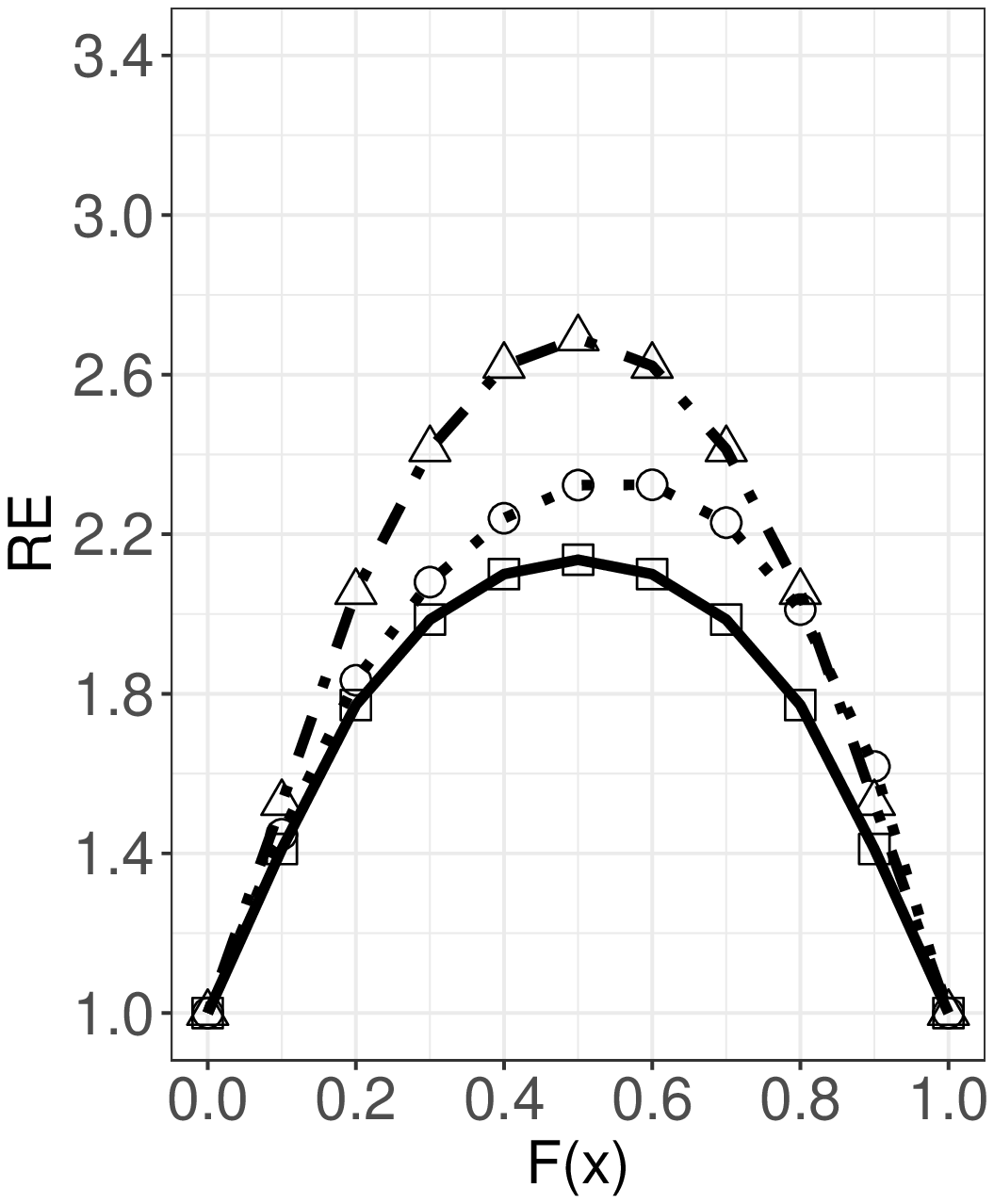}}
	\subfloat[For $k=7$.]{\includegraphics[width=6.8cm,height=5.6cm]{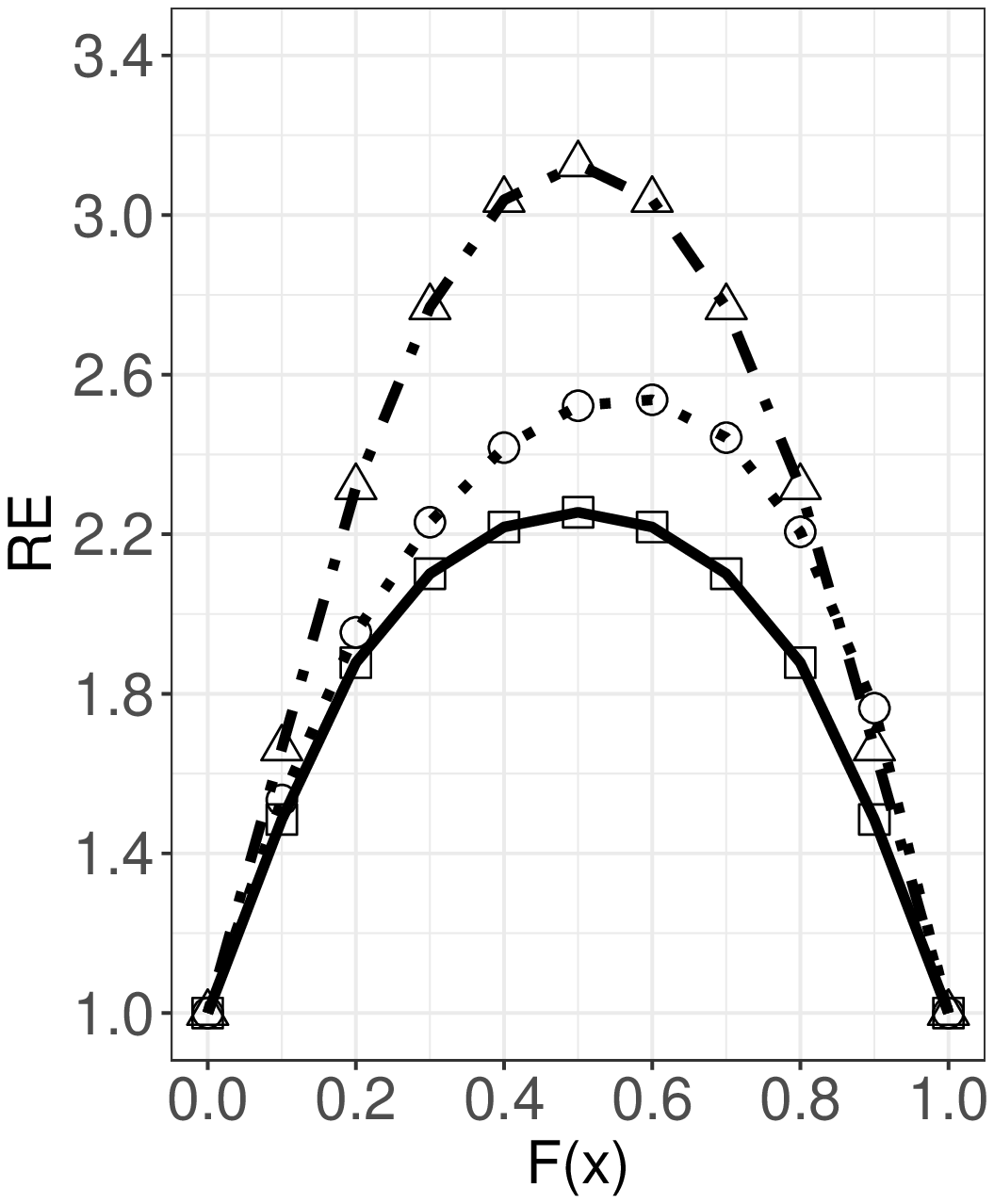}}
	\caption{For $N=50$, the values of $RE(\hat{F}_{h}(x_{p}),\hat{F}_{L-t}^{\ast }(x_{p}))$ for $t=0,1,2$ where $F(x)=p$}
	\label{fig2}
\end{figure}
\begin{figure}[http!]
	\centering
	\subfloat[For $m=2$ and $k=3$.]{\includegraphics[width=6.8cm,height=7.5cm]{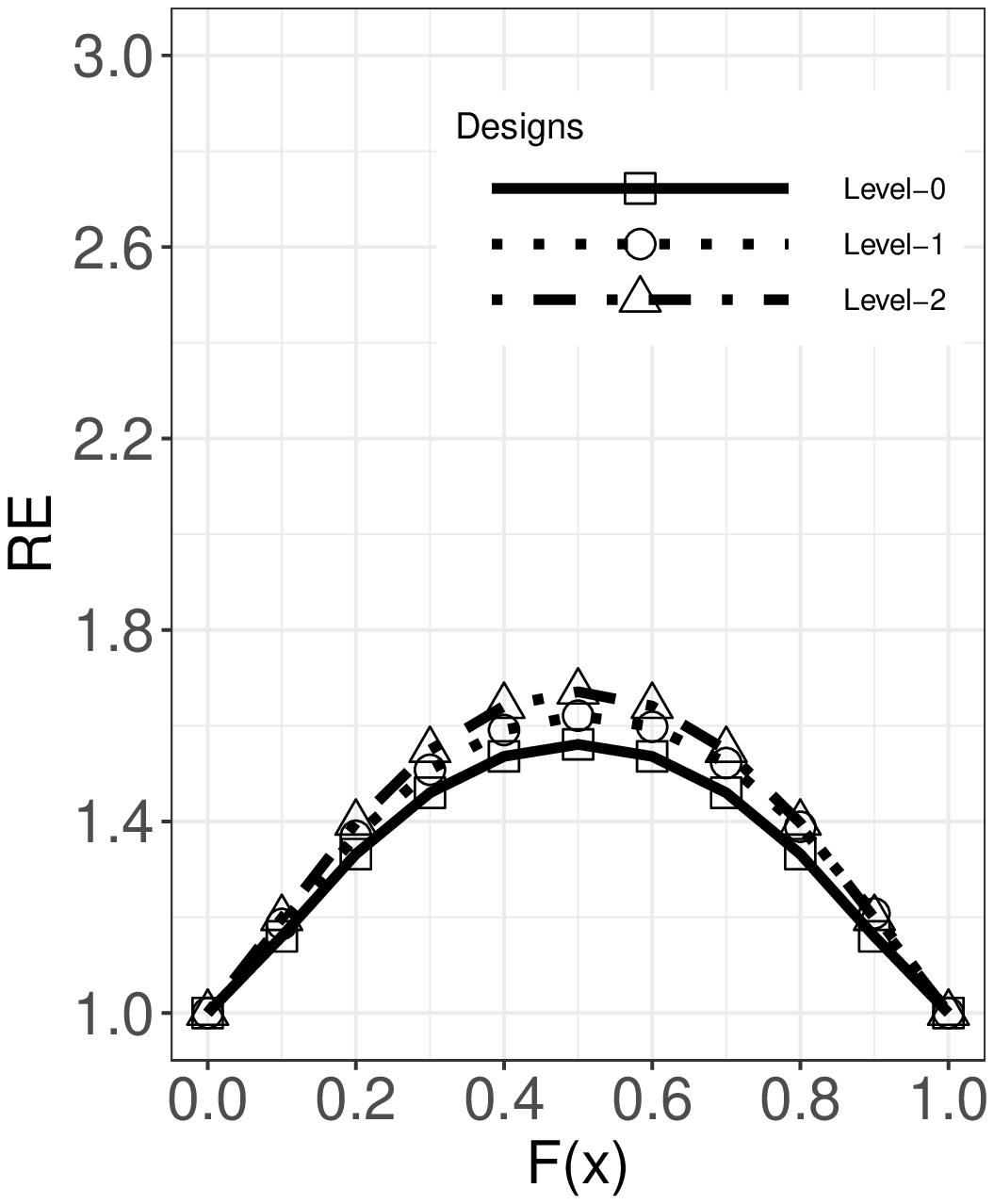}}
	\subfloat[For $m=2$ and $k=5$.]{\includegraphics[width=6.8cm,height=7.5cm]{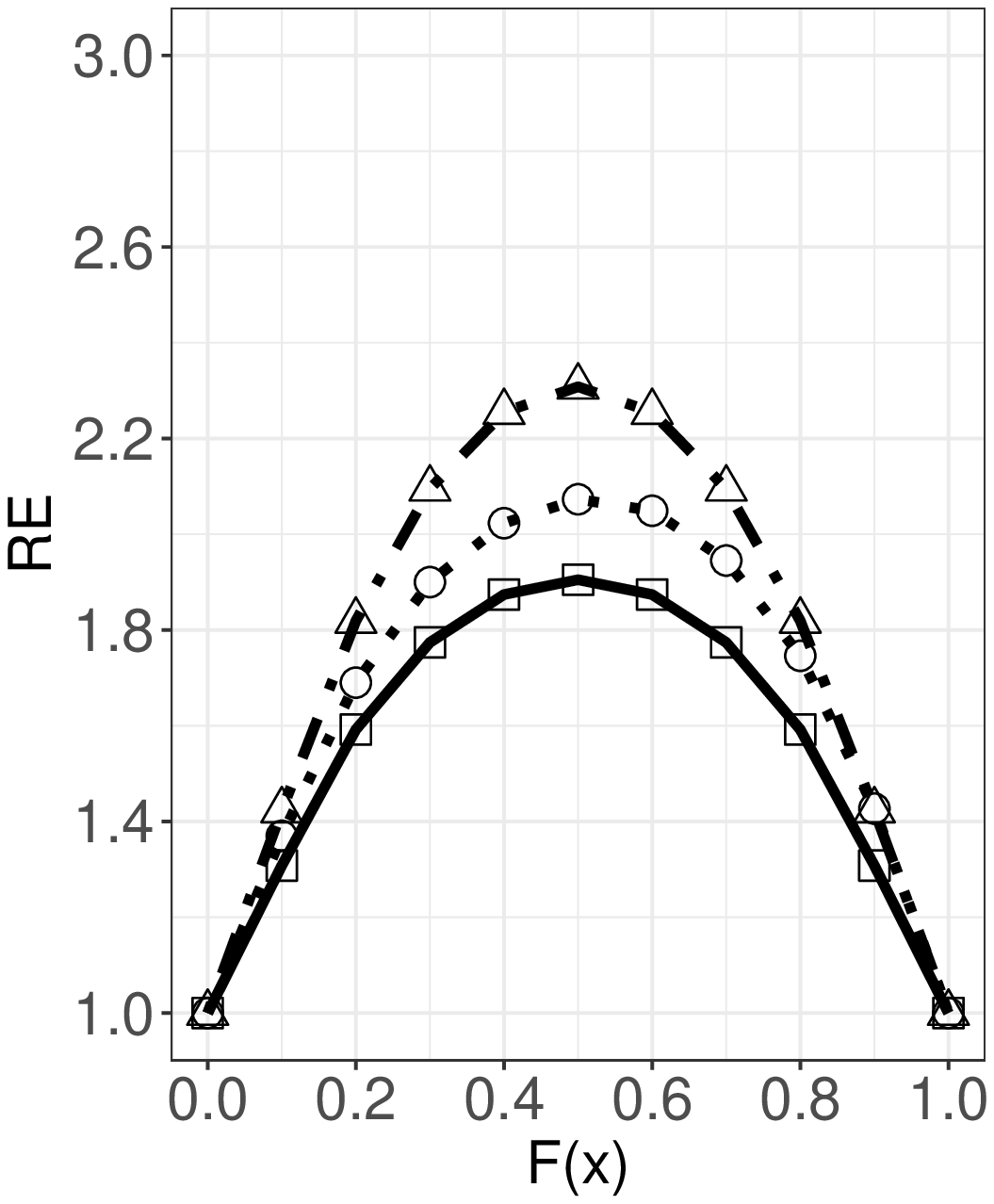}}
	\hskip2em
	\subfloat[For $m=4$ and $k=3$.]{\includegraphics[width=6.8cm,height=7.5cm]{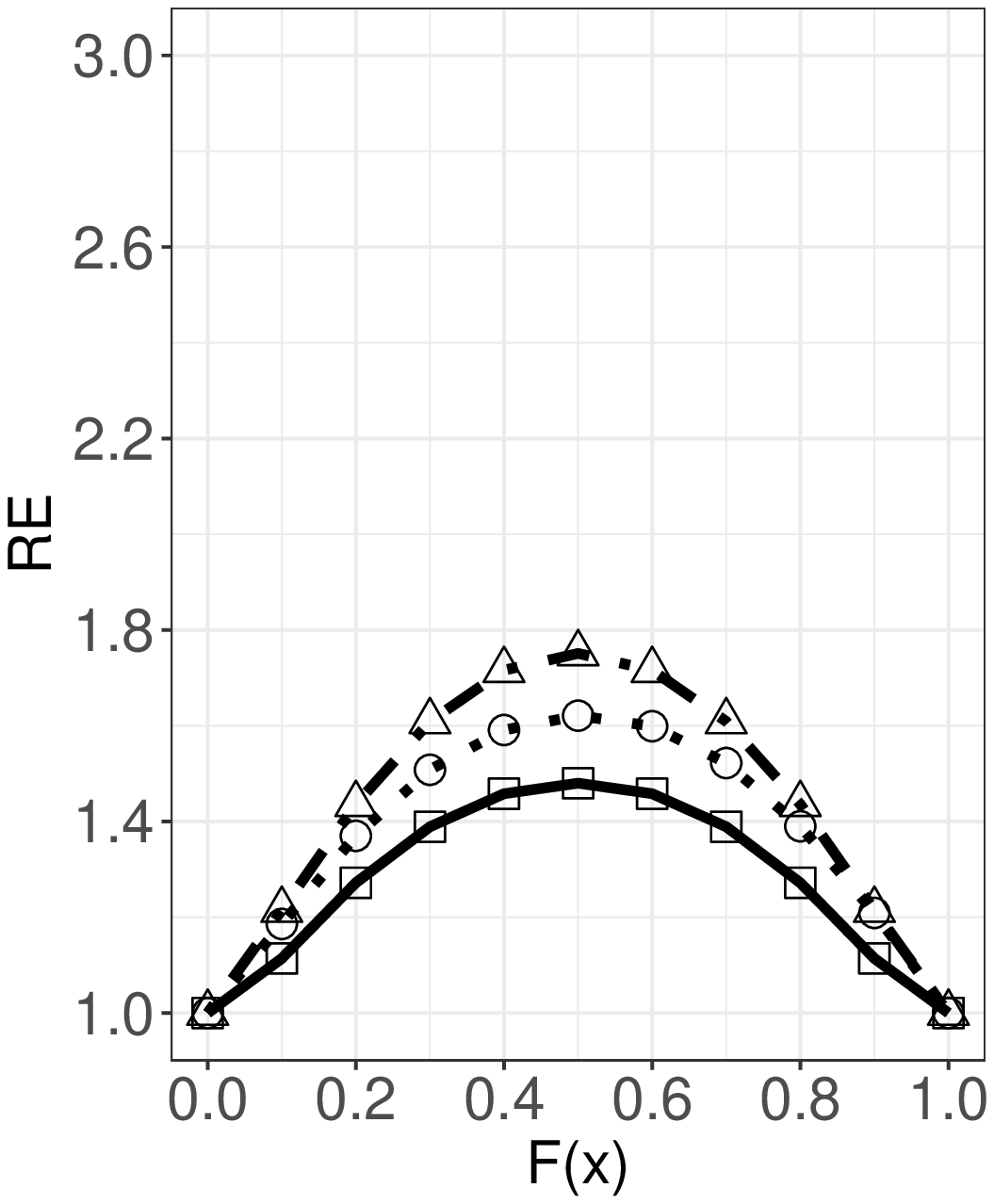}}
	\subfloat[For $m=4$ and $k=5$.]{\includegraphics[width=6.8cm,height=7.5cm]{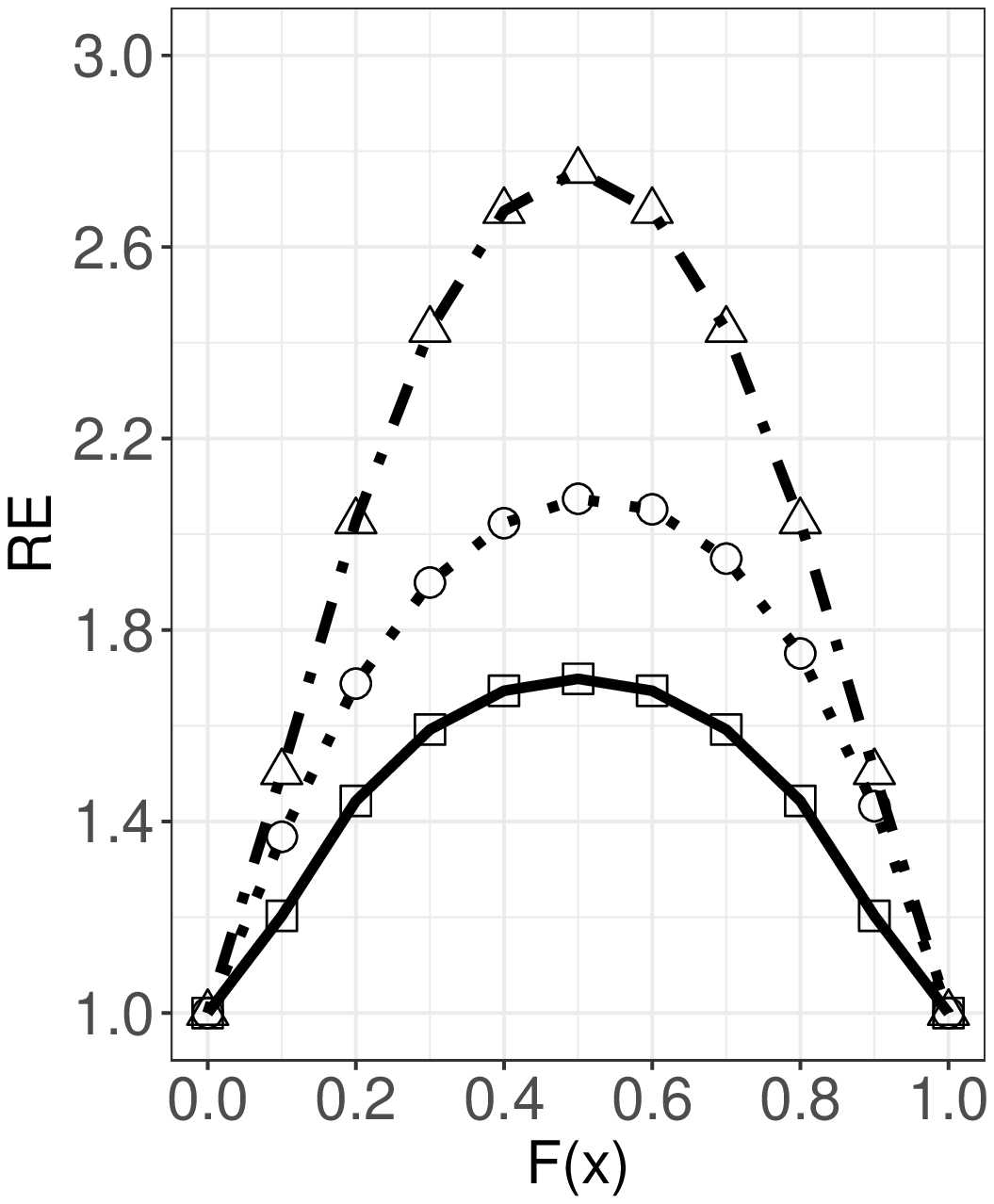}}
	\caption{For $N=100$, the values of $RE(\hat{F}_{h}(x_{p}),\hat{F}_{L-t}^{\ast }(x_{p}))$ for $t=0,1,2$ where $F(x)=p$}
	\label{fig3}
\end{figure}
However, this superiority stems from lack of symmetry in the level-1 sampling process. Because the order in which the observations are collected matters. This could be a factor and this superiority can be ignored. The estimator $\hat{F}_{L-2}^{\ast }(x)$ has highest performance when $N=100$ for any values of $F(x)=p$. The efficiency values of the proposed estimators are increasing functions of the set size $k$ and the number of cycles $m$. Moreover, it is clearly seen that increase in the set size $k$ and the number of cycle $m$ gain advantage especially for the EDF based on level-2. Because level-2 is obtained without replacement policy. For example, Figure \ref{fig2} indicates that the RE of the EDF based on level-2 is $1.4$ for $k=2$ and $F(x)=0.5$ while the RE is $3.2$ for $k=7$ and $F(x)=0.5$. On the other hand, the RE of the EDF based on level-0 is $1.4$ for $k=2$ and $F(x)=0.5$ while the RE is $2.2$ for the set size $k=7$ and $F(x)=0.5$. According to Figure \ref{fig3}, it can be said that the increase in set size has a greater effect on the relative efficiency than the increase in the number of cycles. For example, the RE of level-2 is $1.8$ for $m=4$, $k=3$ and $F(x)=0.5$ while the RE is $2.8$ for $m=4$, $k=5$ and $F(x)=0.5$. However, the RE of the estimator based on level-2 is $1.6$ for $m=2$, $k=3$ and $F(x)=0.5$ while the RE is $1.8$ for $m=4$, $k=3$ and $F(x)=0.5$. Note that the same REs values can be obtained regardless of the distribution of population. Because our calculations figure out that the distribution which is used has no effect on REs under perfect ranking case.

\section{Imperfect ranking}

Even if the theoretical background of the proposed estimators in the perfect ranking case has been examined, the performance of the new EDF estimators should also be examined for the imperfect ranking case since the perfect ranking assumption is not realistic in practice. 

Ranking procedure is usually performed through subjective judgement or single auxiliary variable. Let us give an example. Assumed that five sheep is selected without replacement among $224$ sheep. Here, the problem is to rank the five sheep according to their weights from the smallest to the largest. The sheep must be ranked without measurement since it is difficult to take measurement from a sheep. If an expert assign ranks ($1\leq r\leq 5$) to the sheep, the ranking quality depends on his/her knowledge about the sheep. On the other hand, an auxiliary variable such as mother's weights or the weights of selected sheep at birth can be used instead of subjective judgement ranking. In this case, the quality of ranking depends on the magnitude of the correlation coefficient between sheep's weights and the selected auxiliary variable. Another important issue is to assign ranks to $224$ sheep to calculate the first and second order inclusion probabilities. This process also varies depending on the ranking error. An important note is that the assignment of rank values to the population units and the ranking of units in the sets must be consistent. In other words, let the mother's weight be preferred for the ranking process. Both the assignment of rank values to the population units and the ranking of units in the sets are performed by mother's weights of the sheep. Therefore, it is important to investigate that the performance of the proposed estimators for the case when the ranking is imperfect. 

For this purpose, we construct a Monte Carlo simulation. To simulate small finite populations such as $N=20$ and $N=100$, Frey \cite{frey2011recursive} suggested an idea. According to this idea, the interested variable $x$ generated by setting $x_{\gamma}=Q_{\kappa}\left(\left(\gamma -0.5 \right)/N  \right) $ where $1\leq \gamma \leq N$, $\kappa=1$, $2$, $3$ and $4$ for quantile functions of standard normal ($N(0,1)$), standard uniform ($U(0,1)$), standard exponential ($Exp(1)$) and beta ($Beta(5,2)$), respectively. Note that $Beta(5,2)$ is left skewed distribution. There are two advantages of this idea. First, it is to create reproducible small finite populations from any distribution. The other is to obtain the values of the interested variable from the entire distribution. Thus, a small finite population which represents the preferred distribution can be obtained. 

In the simulation, the other parameters are taken to be $k\in\left\lbrace 2,3,4 \right\rbrace $ for $N=20$, $m\in\left\lbrace 2,4 \right\rbrace$ and $k\in\left\lbrace 3,5 \right\rbrace$ for $N=100$. The ranking procedure is performed by using the following ranking error model.
\begin{equation}
	Y_{\gamma}=\rho \left( \frac{X_{\gamma}-\mu _{x}}{\sigma _{x}}\right) +\sqrt{1-\rho ^{2}}Z_{\gamma}
\end{equation}
where $Y$ is the auxiliary variable, $Z$ follows the standard normal distribution and independent from $X$, $1\leq\gamma\leq N$. Note that the values of $Z$s cannot be generated by using the reproducible way. Because, a strong correlation between $X$ and $Y$ arises. Different way may be considered for this problem in further studies. However, we can say that this problem has little effect on the results and these results can be reproduced by the reader in a simulation with $10,000$ iterations. In Eq. (12), the ranking quality is controlled by the magnitude of the correlation coefficient $\rho \in \left[-1,1 \right] $. In the simulation, the correlation coefficient is taken to be $\rho \in \left\lbrace 0.9,0.75,0.5\right\rbrace $ where $\rho=0.9$ means nearly perfect ranking, $\rho=0.75$ means imperfect ranking which is good enough and $\rho=0.5$ means imperfect ranking. In this simulation, we have generated $10,000$ random samples from SRS and level-t sampling design where $t=0,1,2$.

The RE of EDF in the level-t sampling design to its counterpart in SRS is computed using 
\begin{equation}
	RE(\hat{F}_{h}(x_{p}),\hat{F}_{L-t}^{\ast }(x_{p}))=\frac{V\left[ \hat{F}%
		_{h}(x_{p})\right] }{MSE\left[ \hat{F}_{L-t}^{\ast }(x_{p})\right] }
\end{equation}
where 	
\begin{equation*}
	MSE\left[ \hat{F}_{L-t}^{\ast }(x_{p})\right]=\hat{V}\left[ \hat{F}_{L-t}^{\ast }(x_{p})\right] + bias^{2}\left[ \hat{F}_{L-t}^{\ast }(x_{p})\right]
\end{equation*}
for $t=0,1,2$ where $x_{p}$ is the value of quantile corresponding to $F(x)=p\in\left\lbrace 0,0.1,0.2,\cdots,0.9,1 \right\rbrace $. In this section, we give some remarkable results in Figures \ref{fig4}-\ref{fig7}. The other results are provided by Figures S1-S17 of the supplementary material. Figures \ref{fig4}-\ref{fig5} consists of the REs which are obtained for $N=20$ and $k=4$. The EDF based on level-2 sampling design have outperformance in most cases. The shape of REs for $N(0,1)$ and $U(0,1)$ is symmetric. Also, the shape of REs for $Exp(1)$ are left-skewed while the shape of REs for $Beta(5,2)$ is slightly right-skewed.
\begin{figure}[http!]
	\centering
	\subfloat[For $N(0,1)$.]{\includegraphics[width=6.8cm,height=7.5cm]{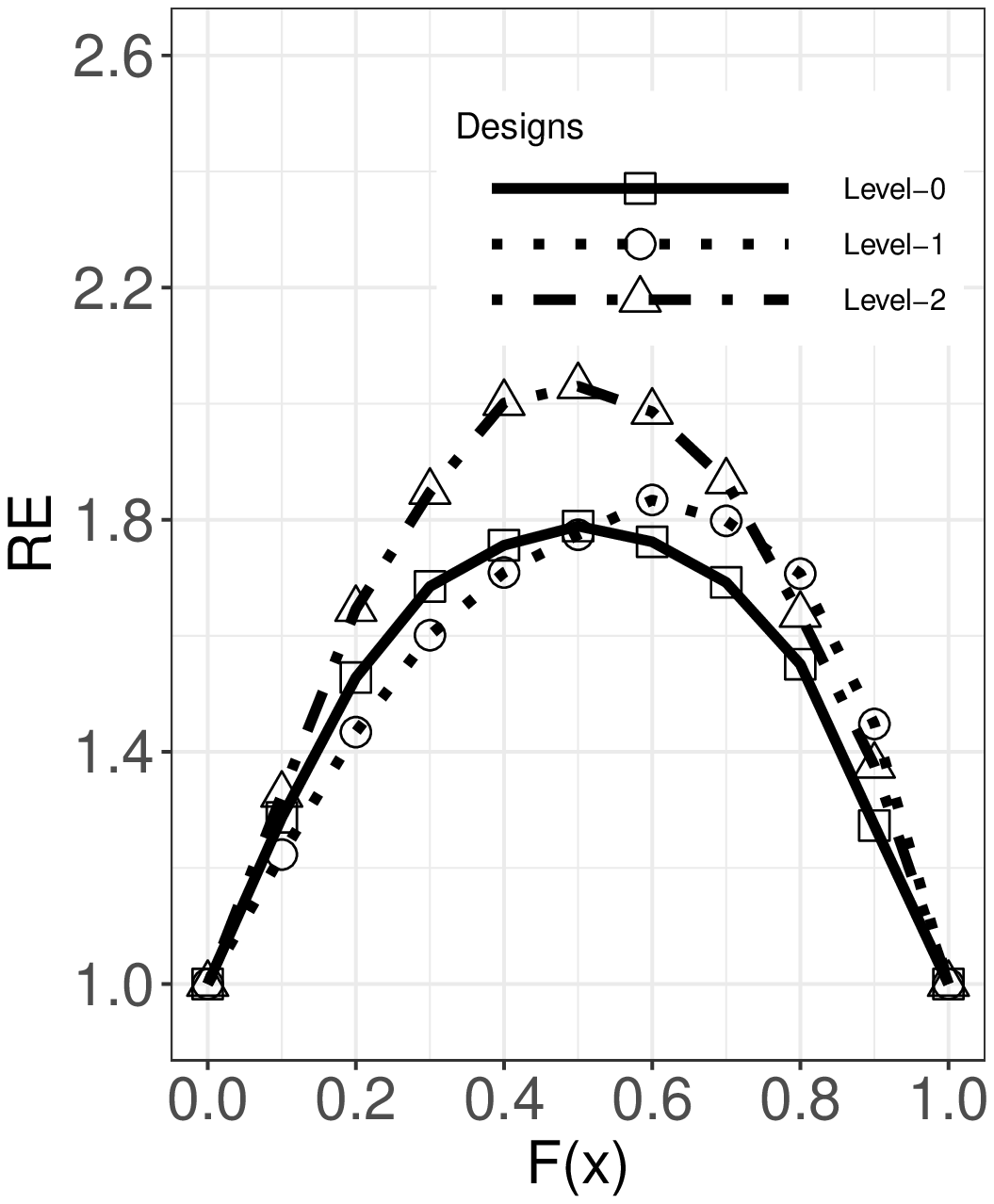}}
	\subfloat[For $U(0,1)$.]{\includegraphics[width=6.8cm,height=7.5cm]{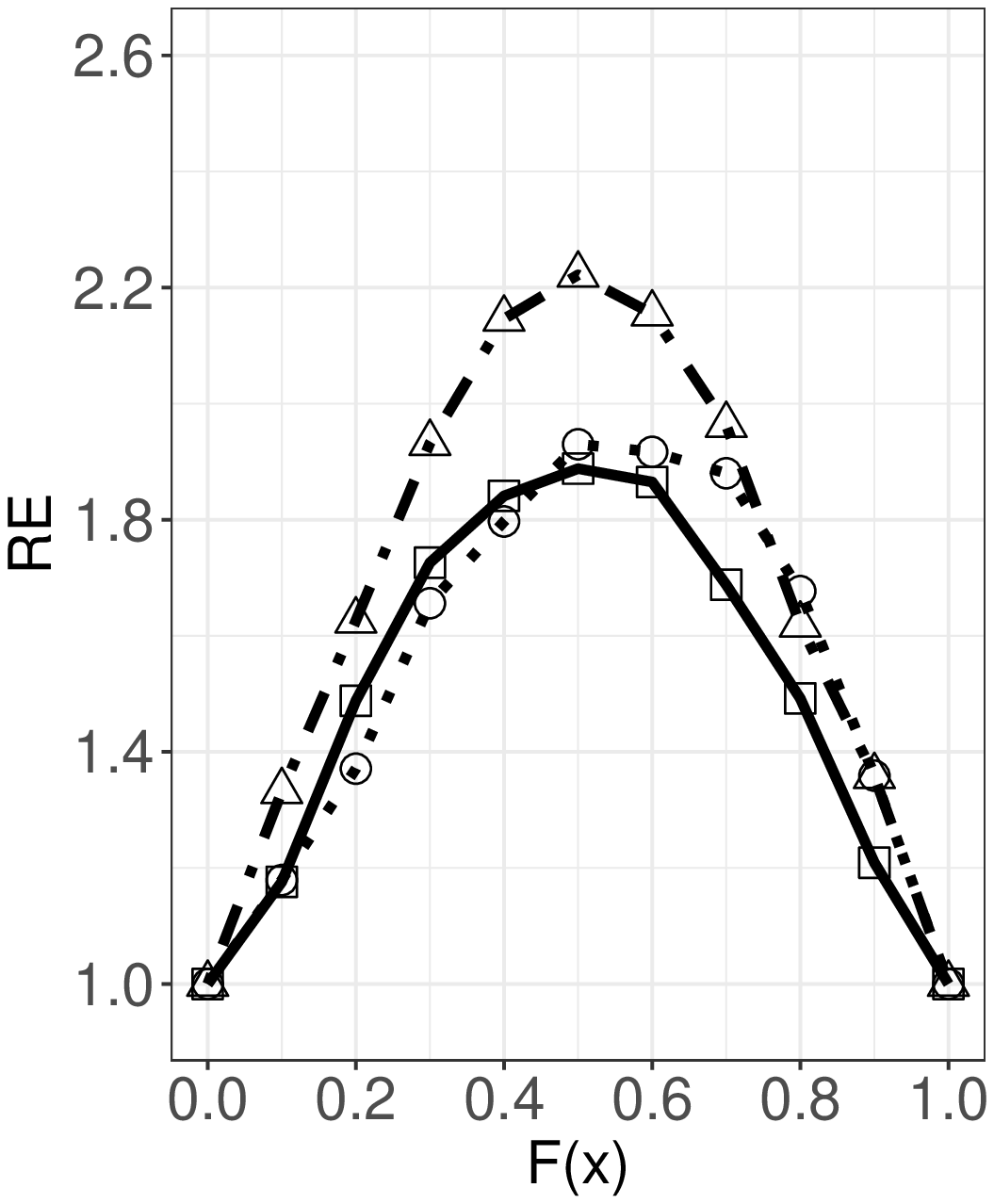}}
	\hskip2em
	\subfloat[For $Exp(1)$.]{\includegraphics[width=6.8cm,height=7.5cm]{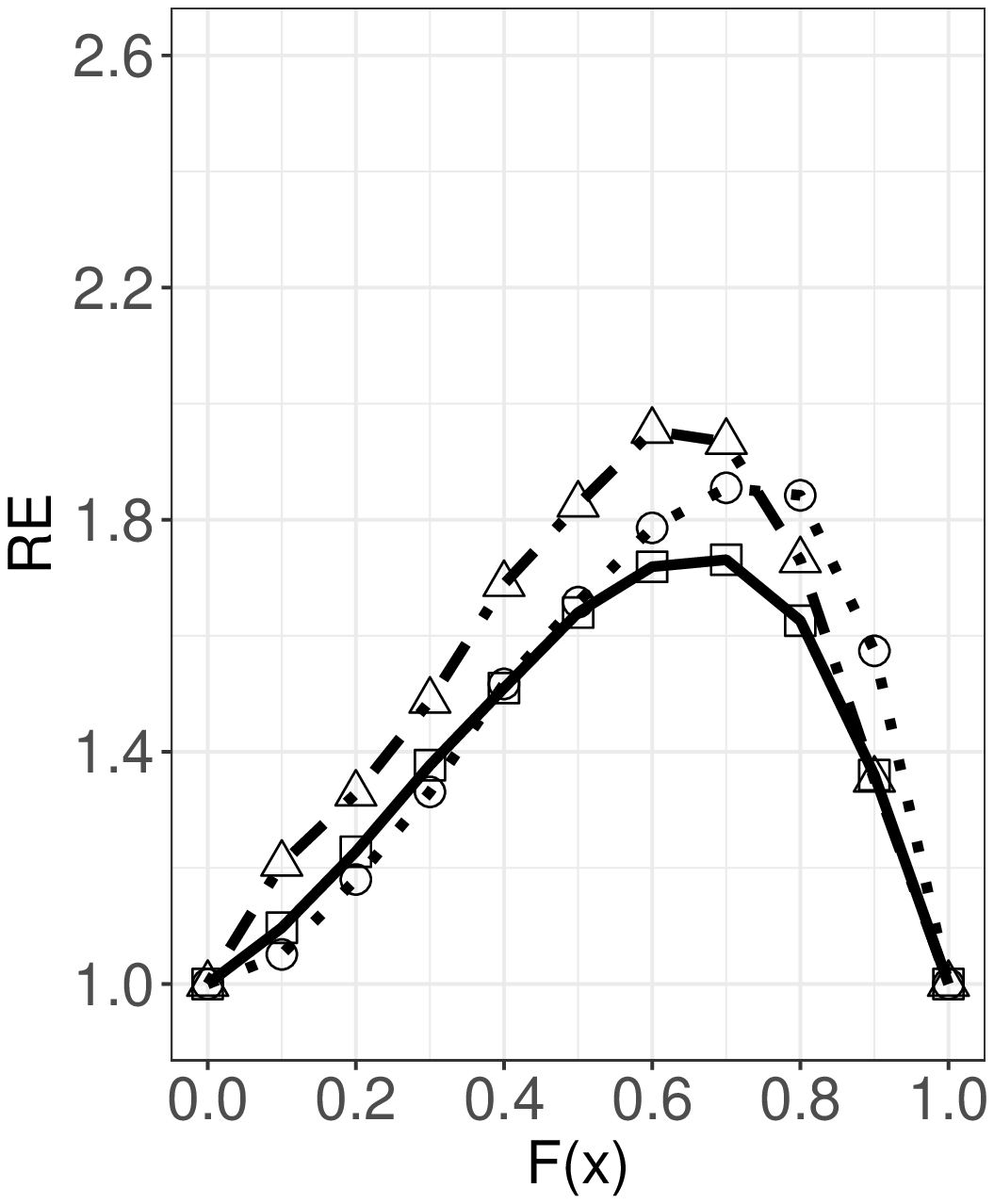}}
	\subfloat[For $Beta(5,2)$.]{\includegraphics[width=6.8cm,height=7.5cm]{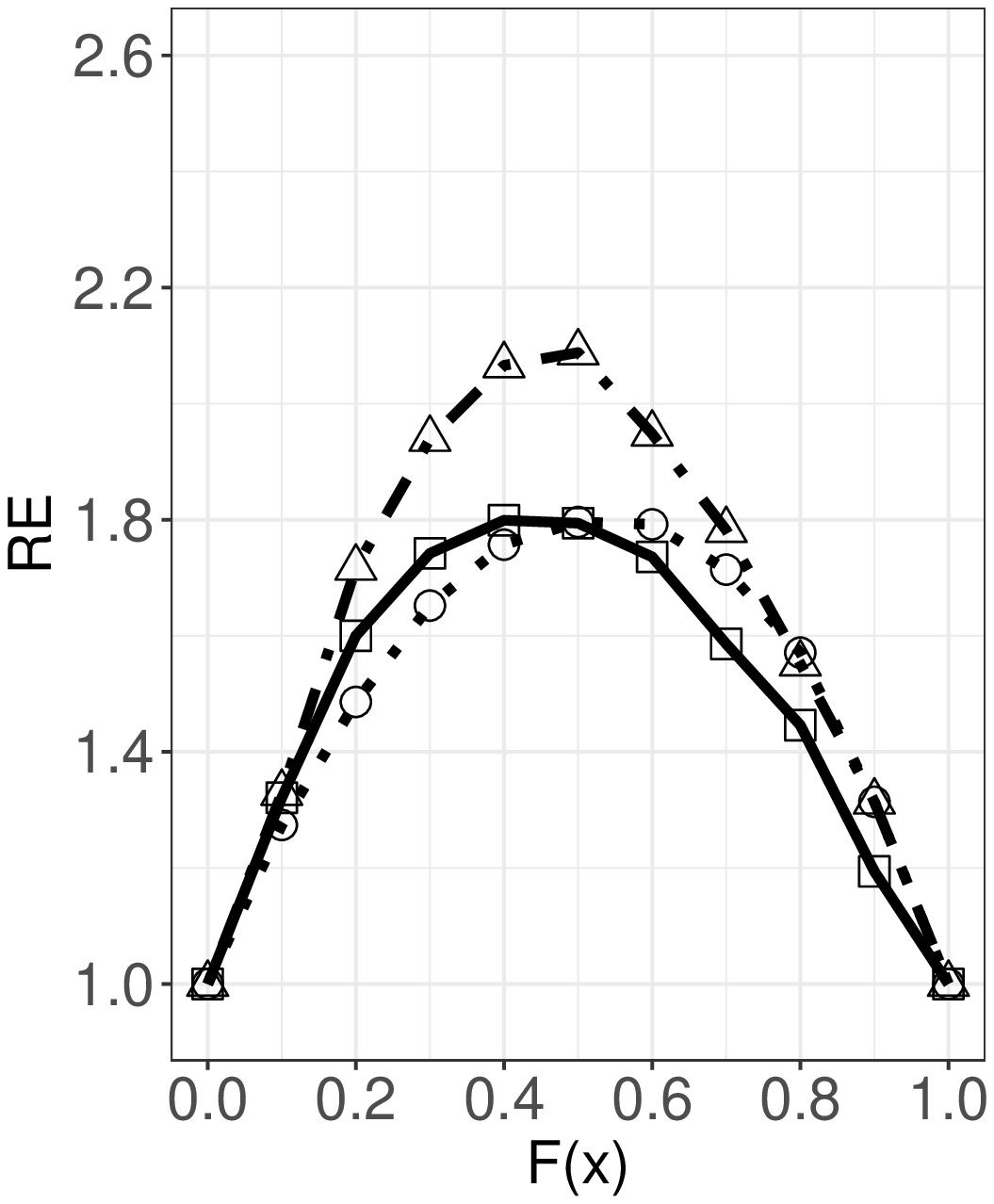}}
	\caption{For the simulation parameters $N=20$, $k=4$ and $\rho=0.9$, the values of $RE(\hat{F}_{h}(x_{p}),\hat{F}_{L-t}^{\ast }(x_{p}))$ for $t=0,1,2$ where $F(x)=p$}
	\label{fig4}
\end{figure}
\begin{figure}[http!]
	\centering
	\subfloat[For $N(0,1)$.]{\includegraphics[width=6.86cm,height=7.5cm]{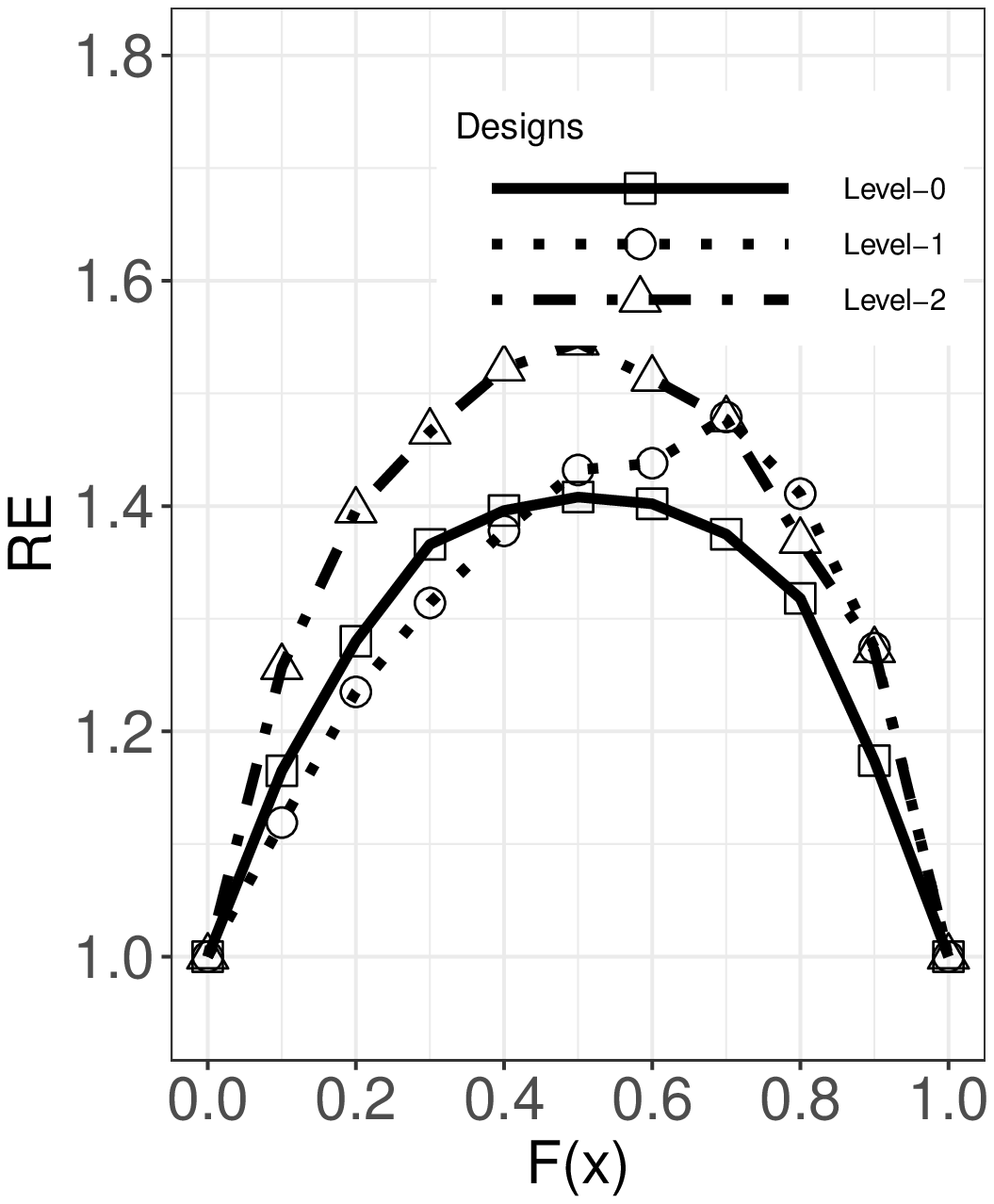}}
	\subfloat[For $U(0,1)$.]{\includegraphics[width=6.8cm,height=7.5cm]{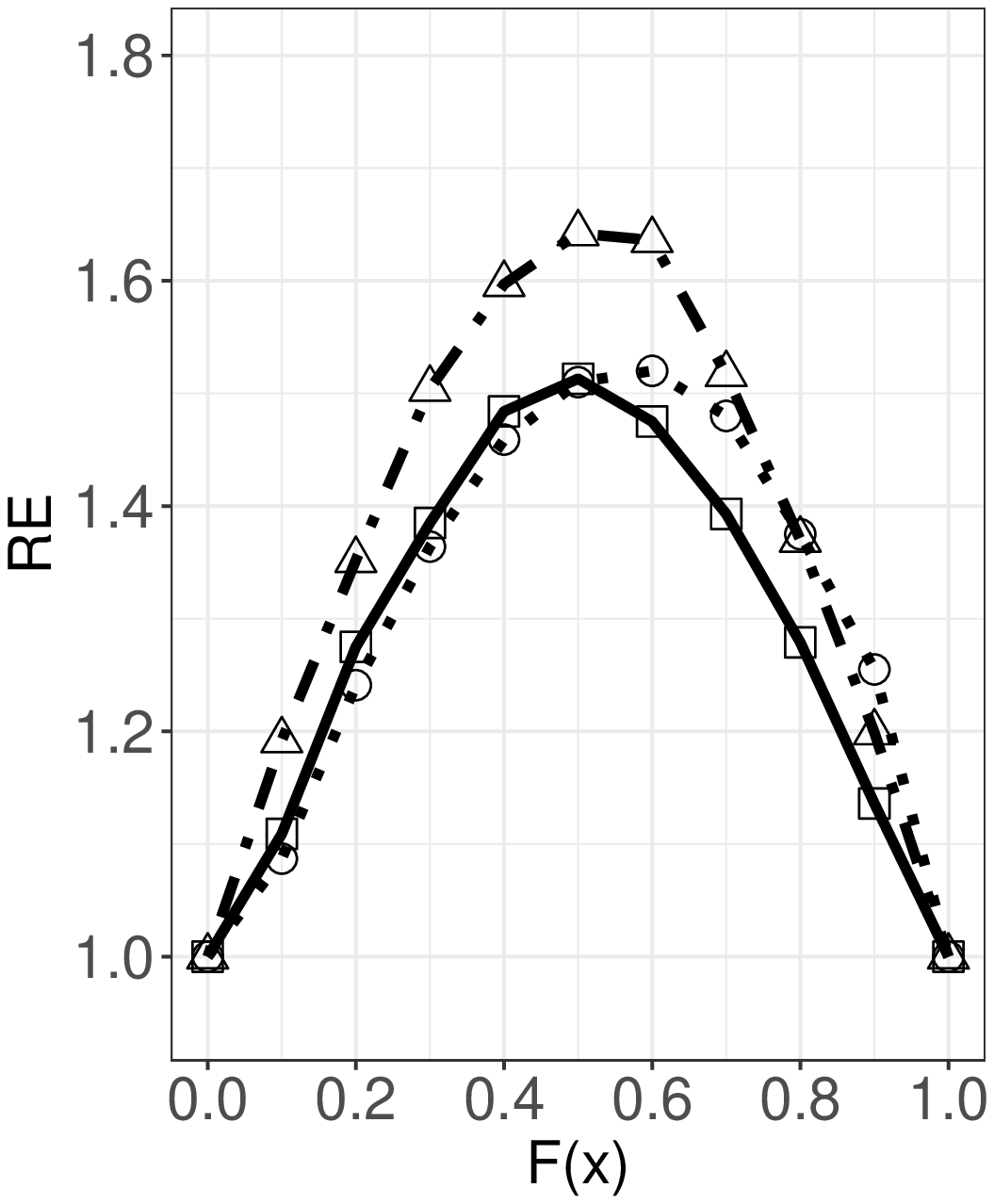}}
	\hskip2em
	\subfloat[For $Exp(1)$.]{\includegraphics[width=6.8cm,height=7.5cm]{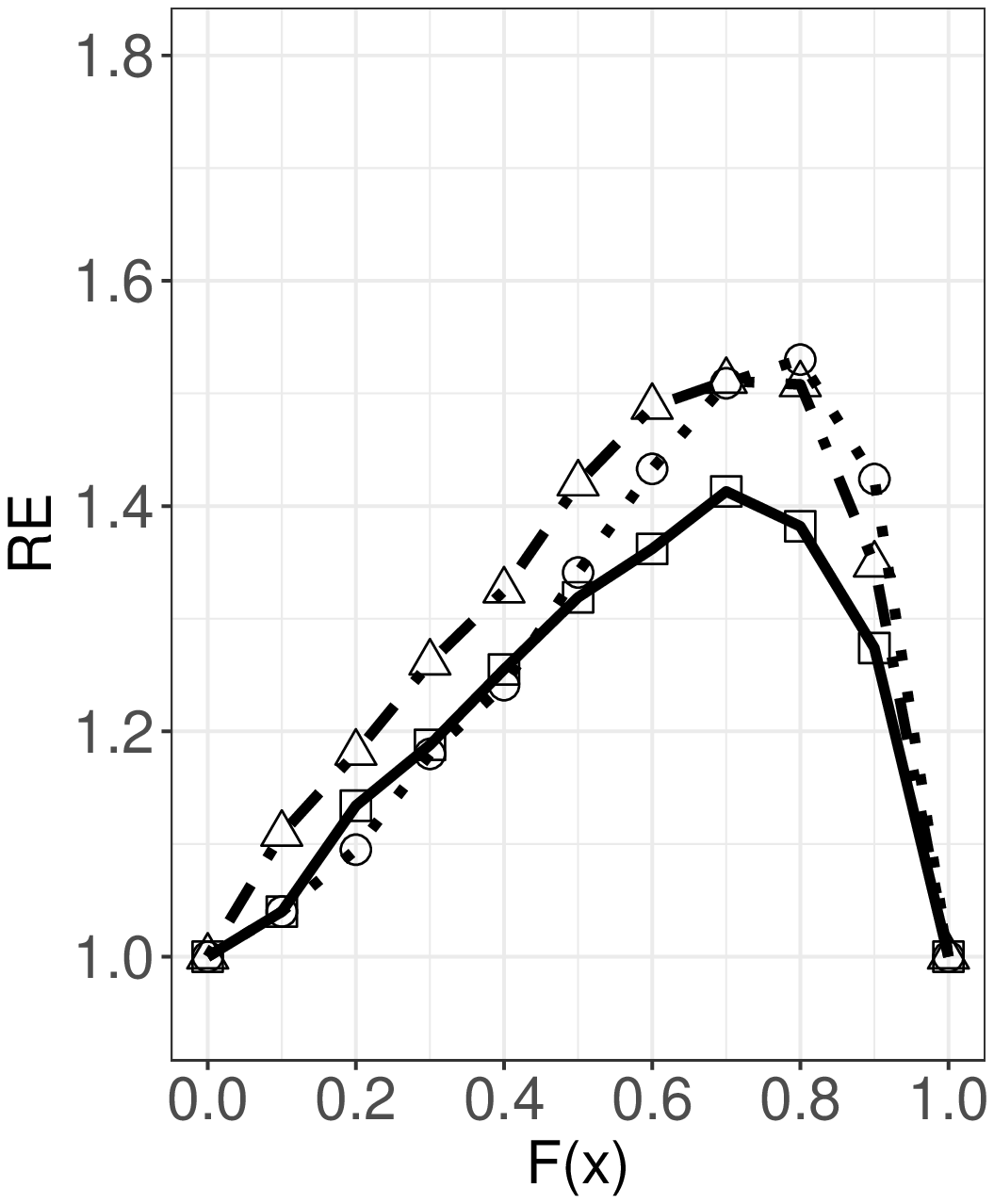}}
	\subfloat[For $Beta(5,2)$.]{\includegraphics[width=6.8cm,height=7.5cm]{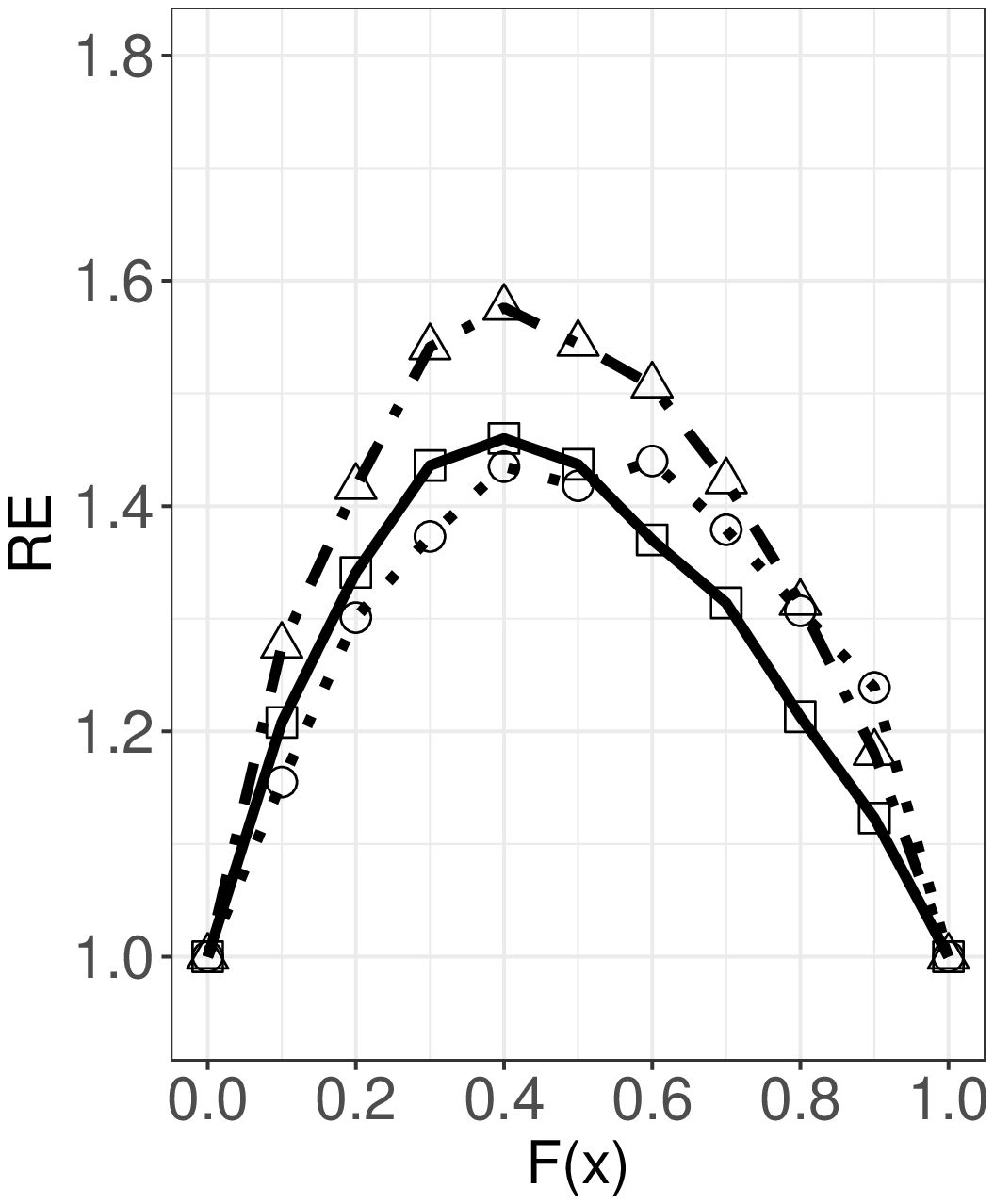}}
	\caption{For the simulation parameters $N=20$, $k=4$ and $\rho=0.75$, the values of $RE(\hat{F}_{h}(x_{p}),\hat{F}_{L-t}^{\ast }(x_{p}))$ for $t=0,1,2$ where $F(x)=p$}
	\label{fig5}
\end{figure}
\begin{figure}[http!]
	\centering
	\subfloat[For $N(0,1)$.]{\includegraphics[width=6.8cm,height=7.5cm]{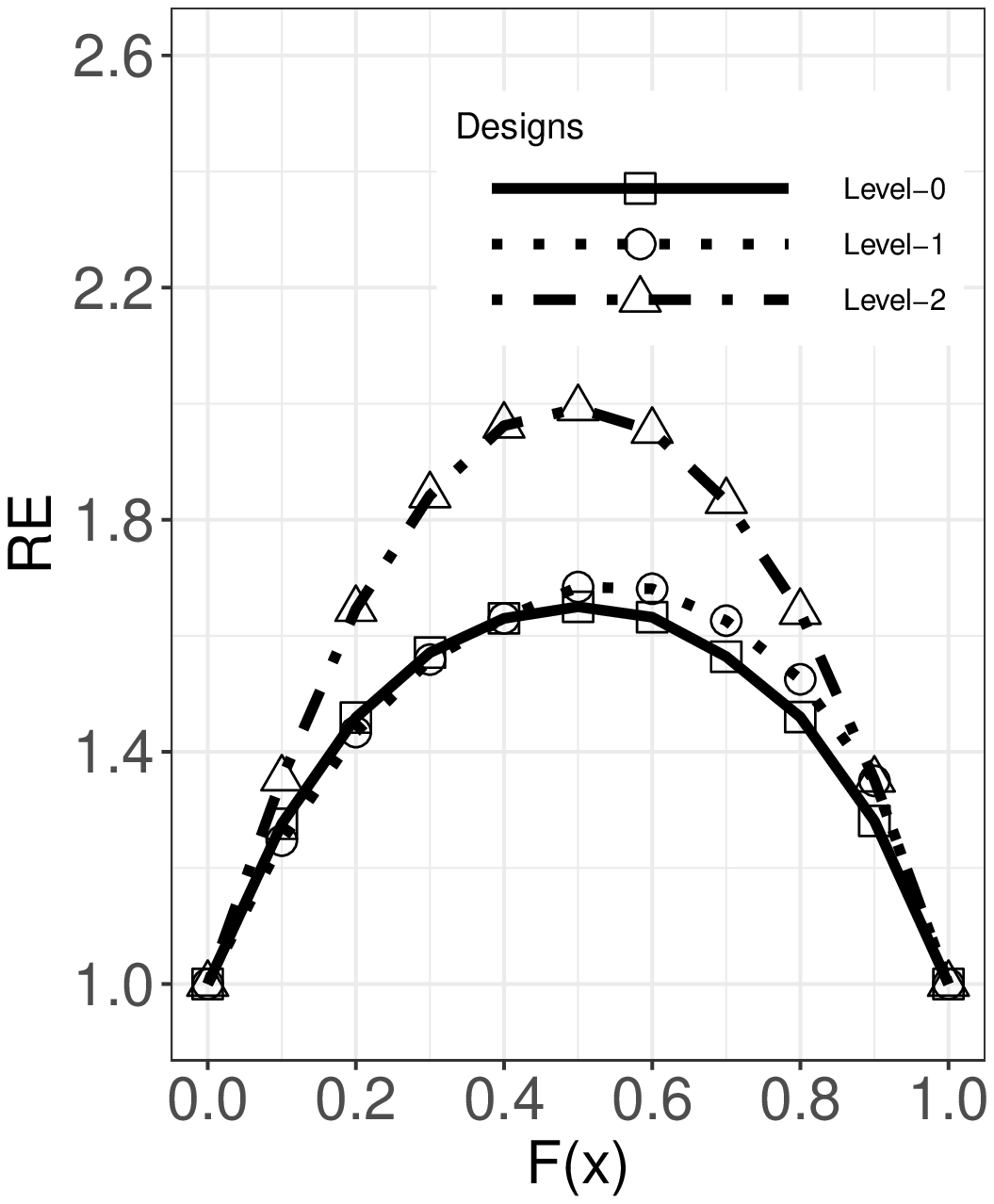}}
	\subfloat[For $U(0,1)$.]{\includegraphics[width=6.8cm,height=7.5cm]{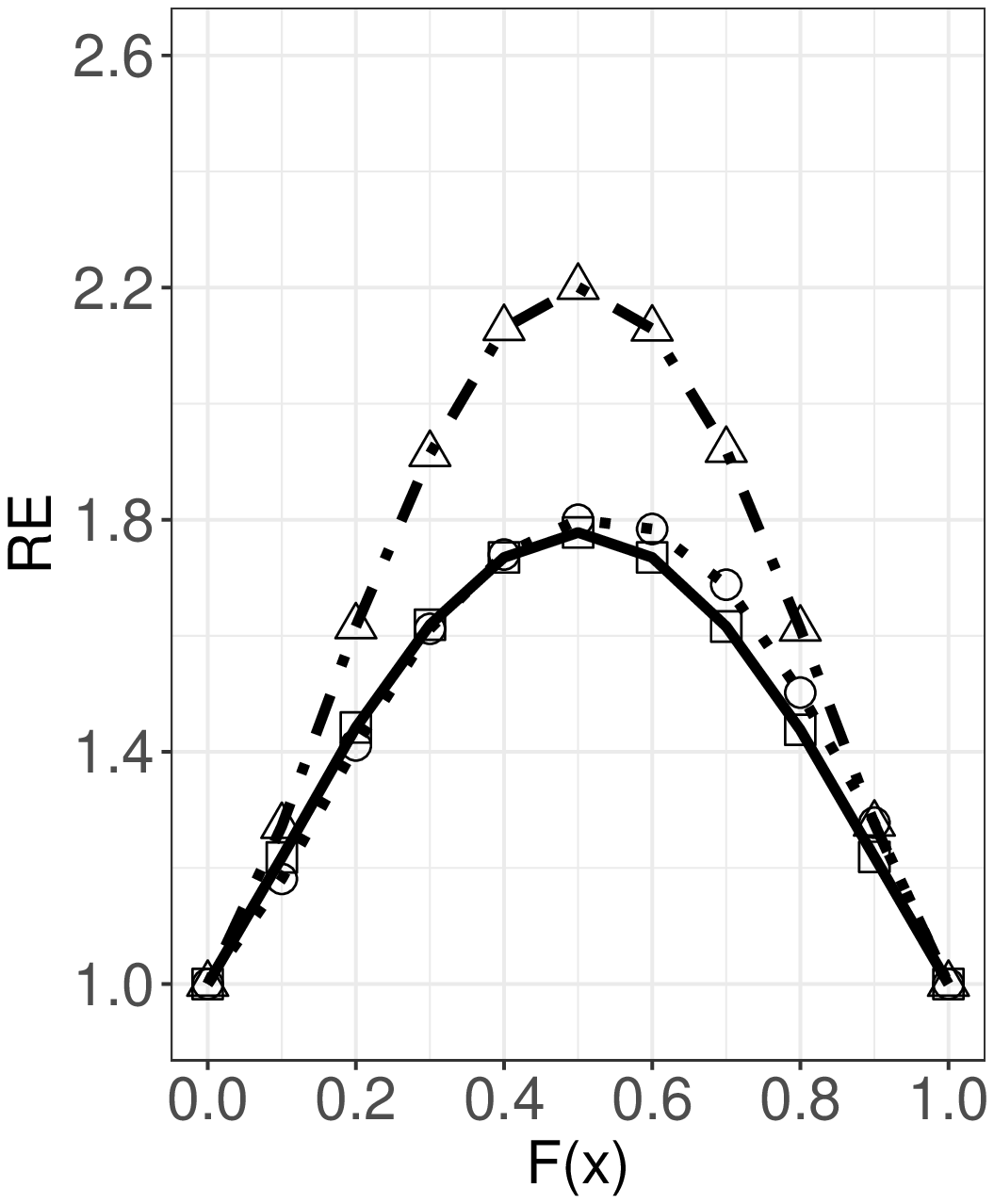}}
	\hskip2em
	\subfloat[For $Exp(1)$.]{\includegraphics[width=6.8cm,height=7.5cm]{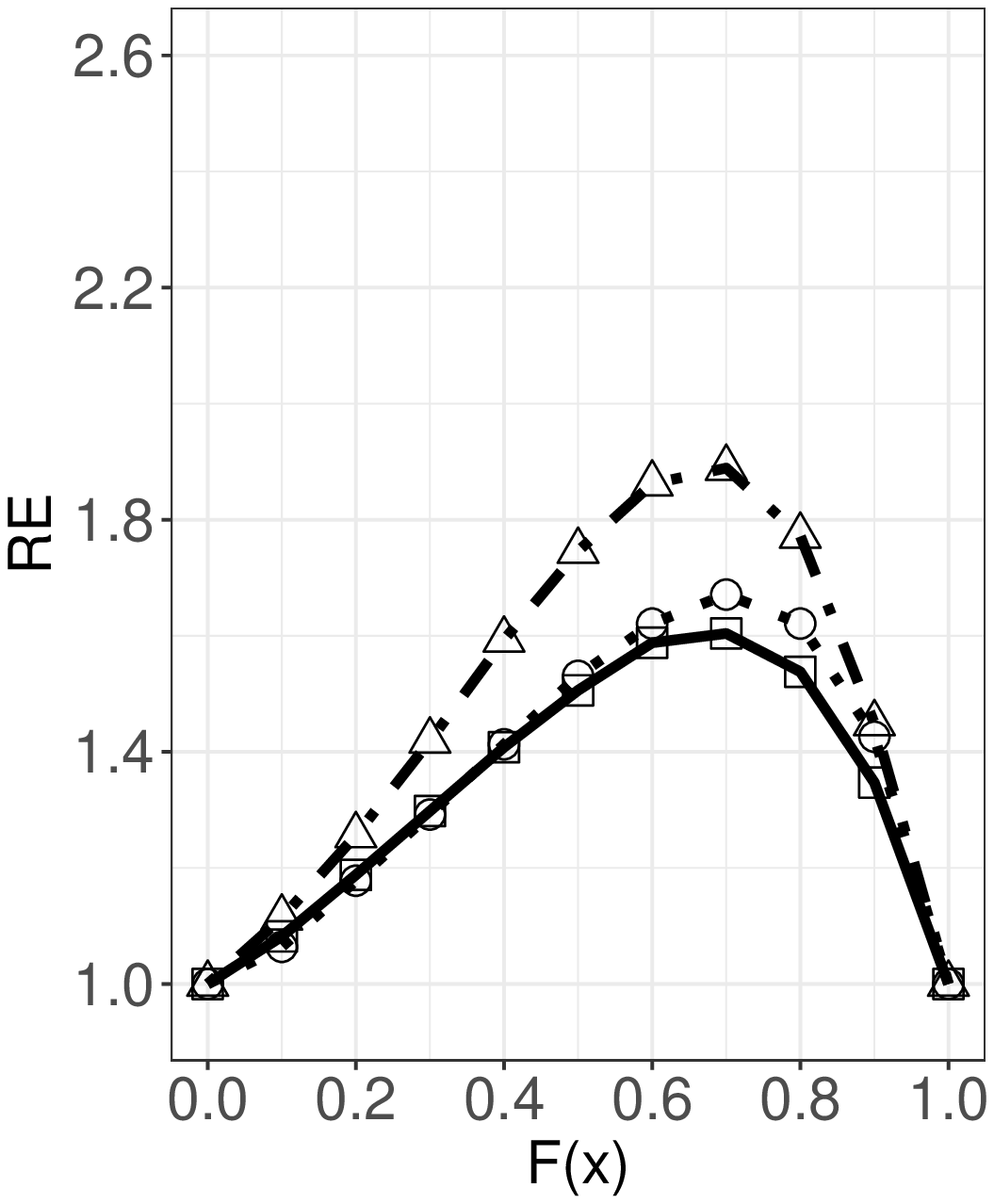}}
	\subfloat[For $Beta(5,2)$.]{\includegraphics[width=6.8cm,height=7.5cm]{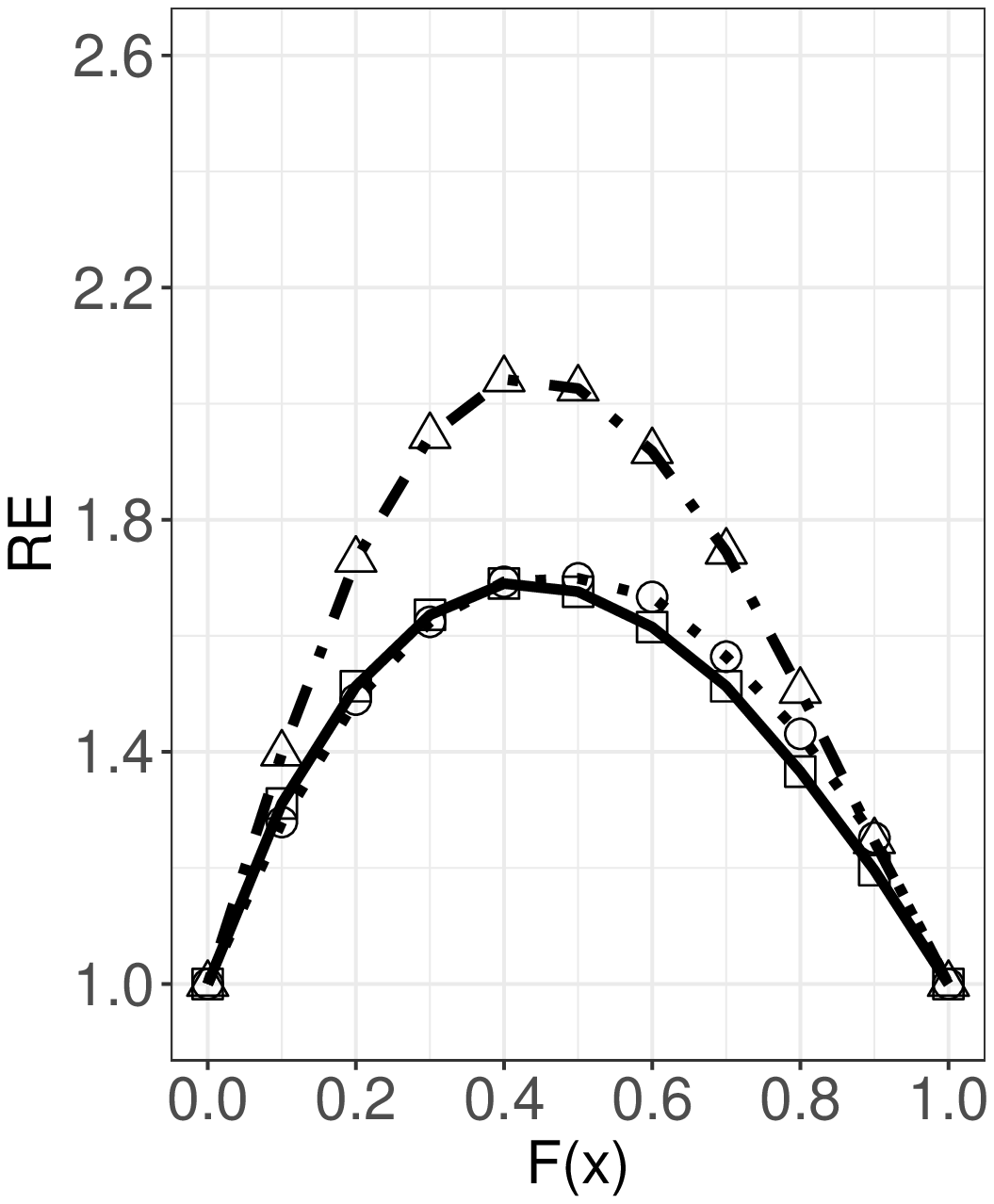}}
	\caption{For the simulation parameters $N=100$, $m=4$, $k=5$ and $\rho=0.9$, the values of $RE(\hat{F}_{h}(x_{p}),\hat{F}_{L-t}^{\ast }(x_{p}))$ for $t=0,1,2$ where $F(x)=p$}
	\label{fig6}
\end{figure}
\begin{figure}[http!]
	\centering
	\subfloat[For $N(0,1)$.]{\includegraphics[width=6cm,height=7.5cm]{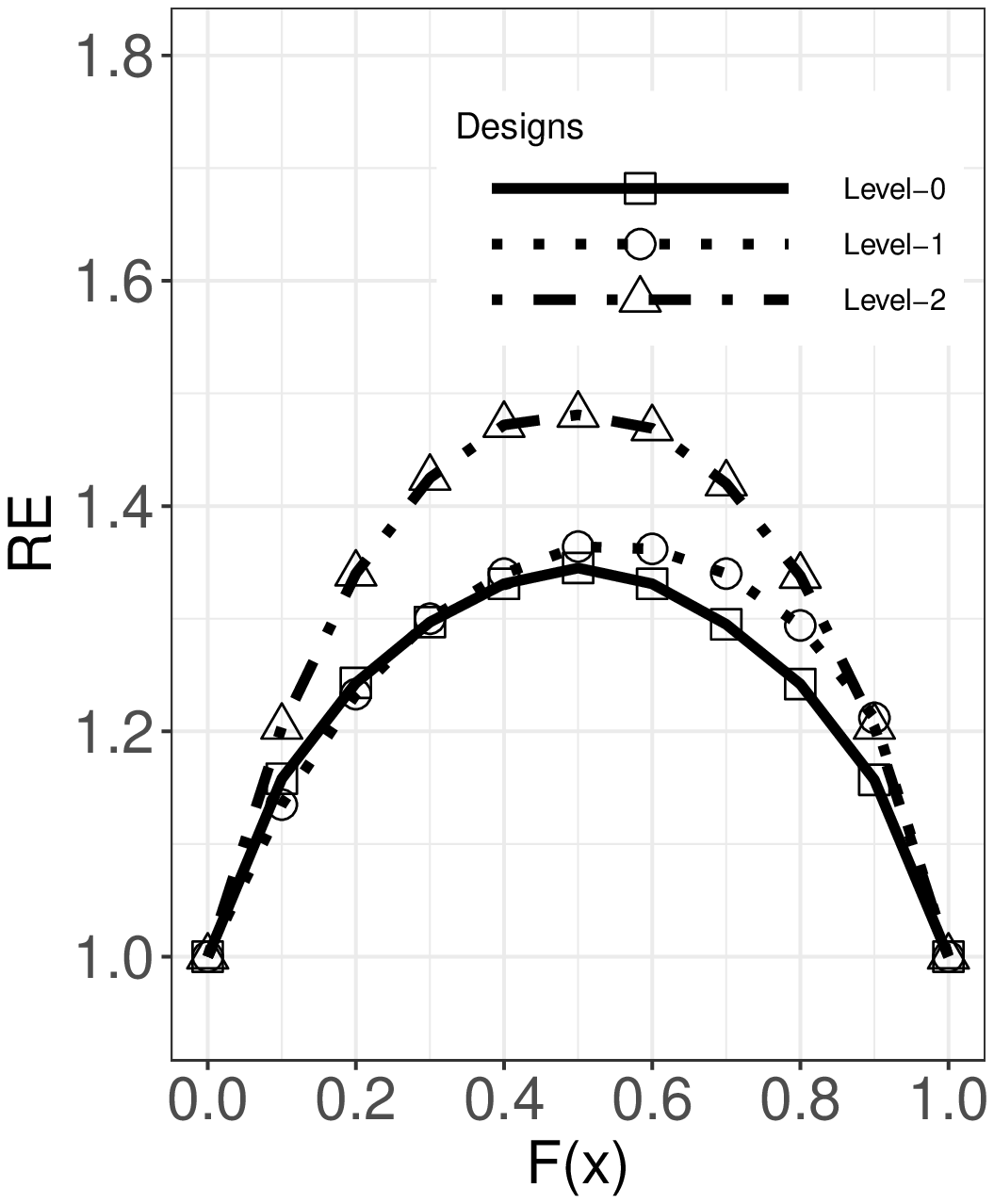}}
	\subfloat[For $U(0,1)$.]{\includegraphics[width=6cm,height=7.5cm]{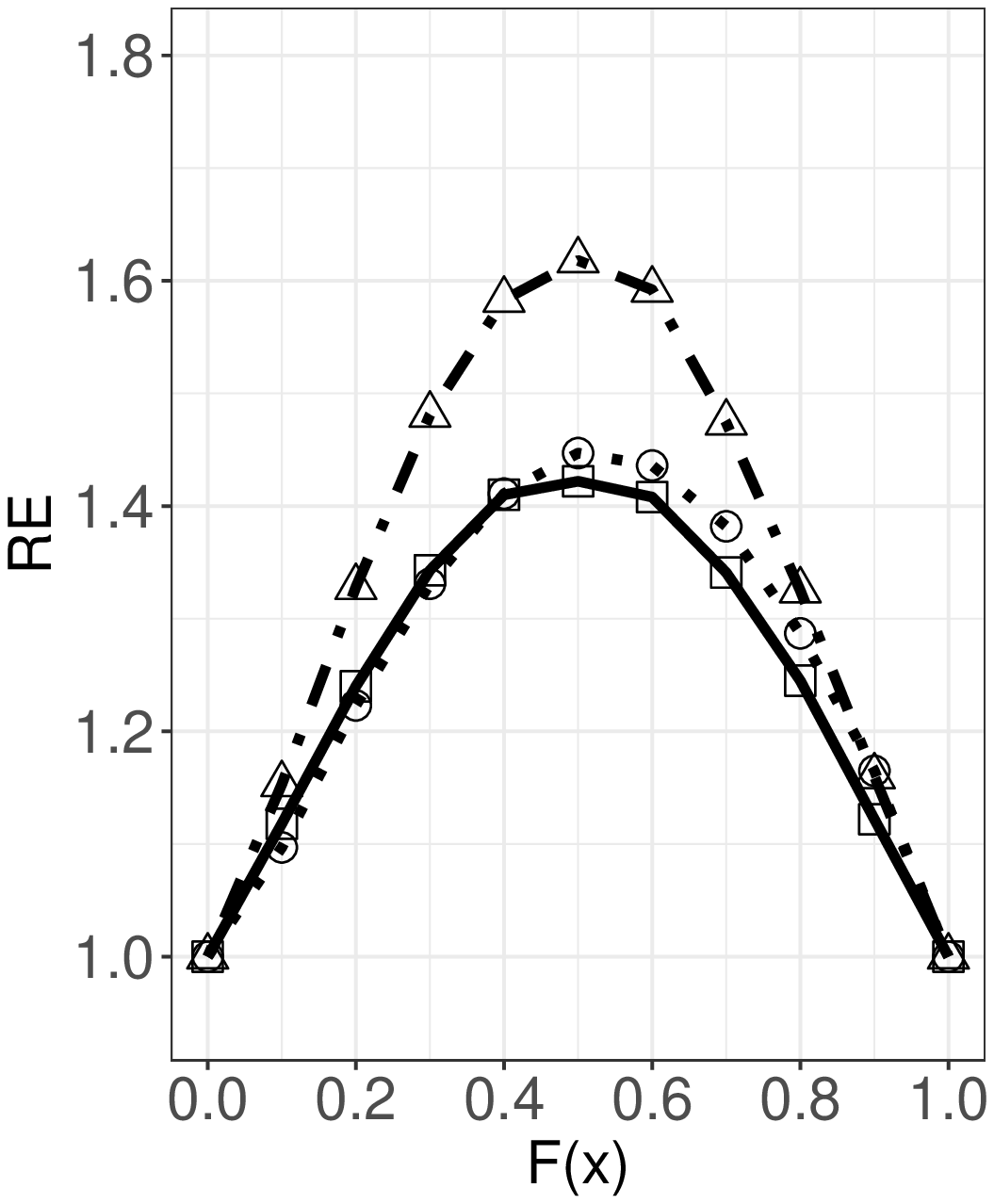}}
	\hskip2em
	\subfloat[For $Exp(1)$.]{\includegraphics[width=6cm,height=7.5cm]{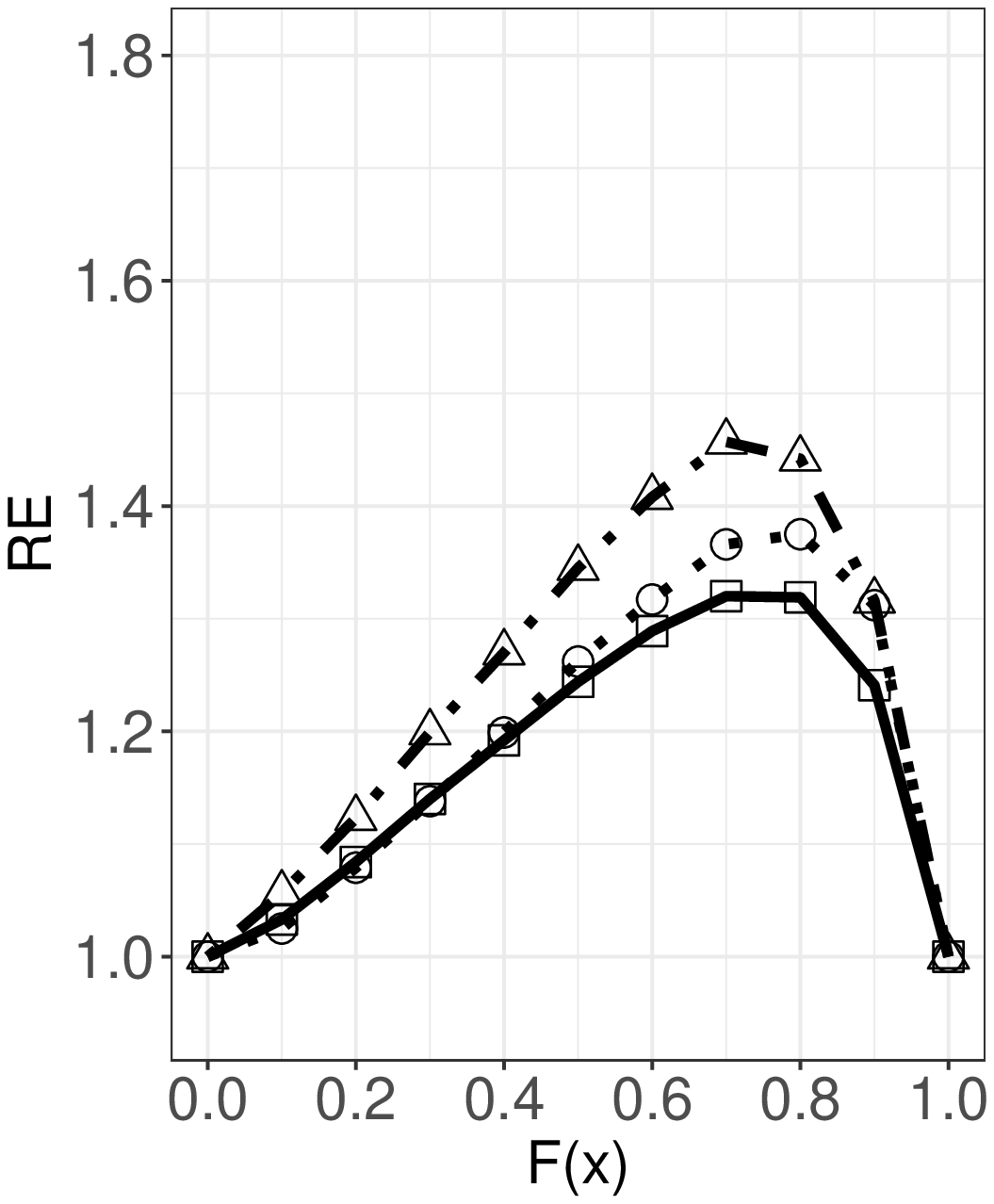}}
	\subfloat[For $Beta(5,2)$.]{\includegraphics[width=6cm,height=7.5cm]{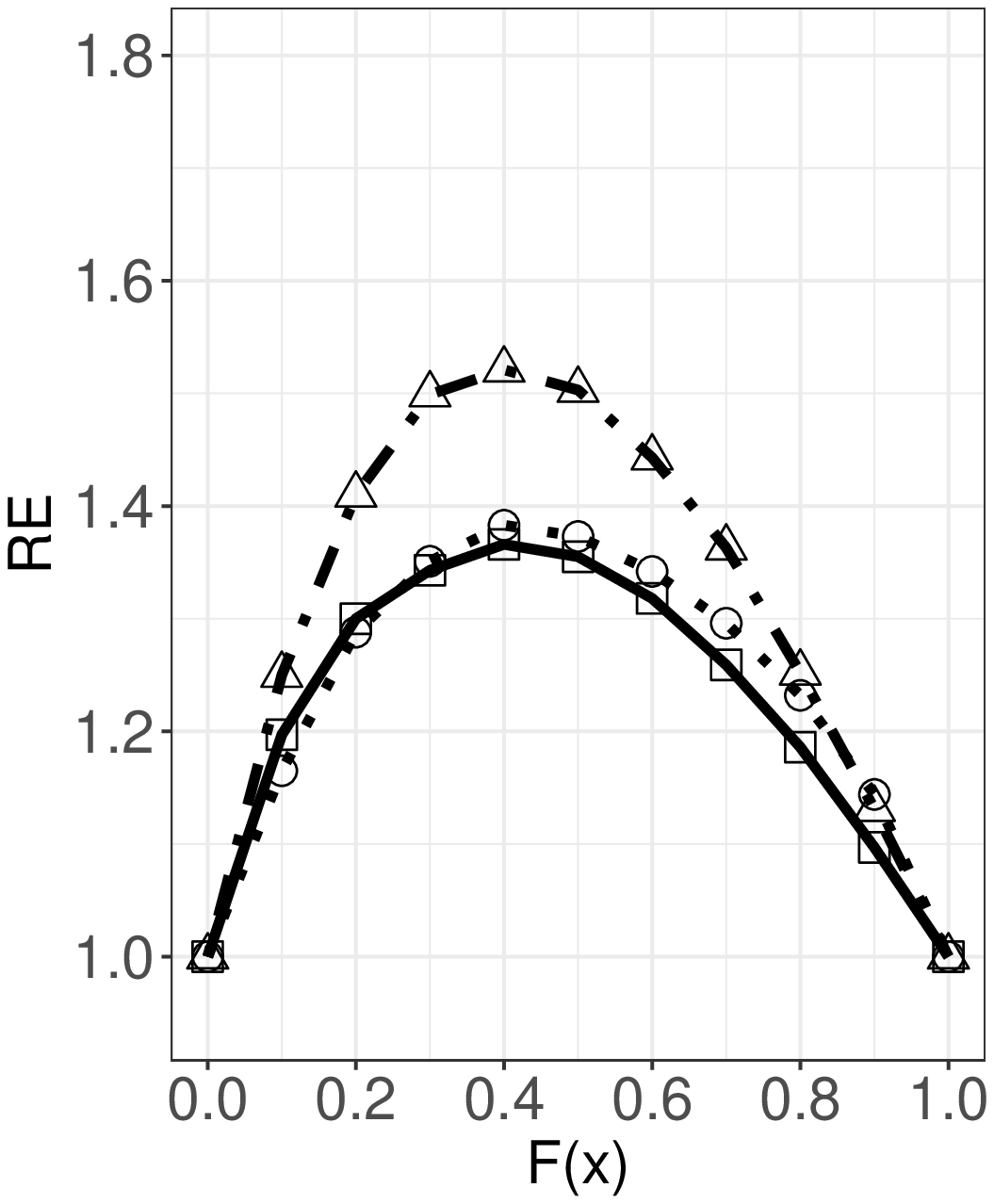}}
	\caption{For the simulation parameters $N=100$, $m=4$, $k=5$ and $\rho=0.75$, the values of $RE(\hat{F}_{h}(x_{p}),\hat{F}_{L-t}^{\ast }(x_{p}))$ for $t=0,1,2$ where $F(x)=p$}
	\label{fig7}
\end{figure}
Moreover, it is seen that the largest RE is obtained for $U(0,1)$ when $F(x)=0.5$. While the correlation coefficient gets closer to $0.75$, the REs decrease as expected. Figure S7 demonstrates that the REs of the EDFs based on sampling designs in RSS becomes poor under the case when the ranking is imperfect. Figures S1-S6 indicate that the RE is increasing in $k$. Also, REs of the EDFs based on level-t are decreasing functions of $\rho$ for each $t=0,1,2$. For Figures \ref{fig6}-\ref{fig7}, we can say that the REs have symmetric distribution around the median of $N(0,1)$ and $U(0,1)$. The shapes of REs for $Exp(1)$ and for $Beta(5,2)$ is left-skewed and slightly right-skewed, respectively. On the other hand, it is appeared that the EDF based on level-2 outperforms among the all EDFs in most cases. Because of the lack of symmetry in the level-1 sampling process, the superiority of the EDF based level-1 can be observed in right tail of the distribution of the REs. However, the superiority can be ignored. Also, we observe that the largest RE is obtained for $U(0,1)$ when $F(x)=0.5$. According to Figures S8-S17 of the supplementary material, it is seen that the REs increase while the set size and/or the number of cycle increase. Moreover, it is clearly observed that increase in the set size rather than increase in the number of cycle has more effect on the REs. For example, the Figures S11 and S14 indicate that REs for $n=10$ ($m=2$ and $k=5$) is larger than the REs for $n=12$ ($m=4$ and $k=3$). In other words, the REs in S11 are obtained under the case $k>m$ while the REs in S14 are obtained under the case $k<m$. Thus, it is an evident that increase in the set size has more effect on REs than increase in the number of cycle. Also, the REs get closer to one while $\rho\rightarrow 0.5$.

\section{Application}

In this section, we apply the EDF estimator based on level-2 sampling design to sheep data since level-2 has outperformance among the three sampling designs and SRS. Here, we explain the sampling procedure and estimation of the distribution function. We aim to show that the sampling procedure is applicable to any environmental or biological data and the proposed estimator can be used easily. 

This data set has been collected by the Research Farm of Ataturk University, Erzurum, Turkey and includes 224 sheep. Aim of the research is to increase meat quality and production. Therefore, a sample is selected periodically to monitor the biological growth and to provide estimates for the population means. The problem is that young sheep are very active animals and it is labor intensive to hold them secure during the measurement process. The measurement errors are mostly appeared because of this active behavior. On the other hand, auxiliary variables such as mother mating weight (kg) and lamb weight (kg) at birth can be used to rank the interested variable, young sheep weight (kg), since these auxiliary variables are accessible or can be measured with less effort than the young sheep. Thus, aim of the present paper is to show that the ranked-based sampling designs can be used effectively to reduce the number of sampled sheep. Also, we show that the sampling designs provide more efficient EDF estimators than the EDF based on SRS for distribution function of sheep's weights at seven months. 

This data set includes three variables which are mother's weight ($\pmb{Y}_{1}$), lamb weight at birth ($\pmb{Y}_{2}$) and sheep weights at seven months ($\pmb{X}$). Note that bold notations are used since each variable is assumed to contain $N$ observations. The correlation coefficients, $Corr(\pmb{X},\pmb{Y}_{1})$ and $Corr(\pmb{X},\pmb{Y}_{2})$, are $0.43$ and $0.79$, respectively. The magnitude of the correlation coefficient determines the quality of ranking. Thus, it is suggested that the lamb weight at birth is used in ranking procedure. The other descriptive statistics are given in Table \ref{Tab1}.
\begin{table}[!htp] 
	\small
	\centering
	\caption{Descriptive statistics}
	\begin{tabular}{l|l|l|l}
		\hline
		& $\pmb{X}$ & $\pmb{Y}_{1}$ & $\pmb{Y}_{2}$ \\ \hline
		Minimum & $20.30$ & $42.20$ & $2.500$ \\ 
		$1$st Quantile & $25.50$ & $49.67$ & $3.875$ \\ 
		Median & $27.90$ & $52.30$ & $4.400$ \\ 
		Mean & $28.11$ & $52.26$ & $4.361$ \\ 
		Standard Deviation & $3.90$ & $4.38$ & $0.792$ \\ 
		$3$rd Quantile & $31.00$ & $55.10$ & $4.800$ \\ 
		Maximum & $40.50$ & $63.70$ & $6.700$ \\ \hline
	\end{tabular}
	\label{Tab1}
\end{table}
Also, Figure \ref{fig10} indicates that the weights of the sheep at seven months has a slightly right skewed distribution. 
\begin{figure}[hbt!]
	\centering
	\includegraphics[width=12cm,height=11cm]{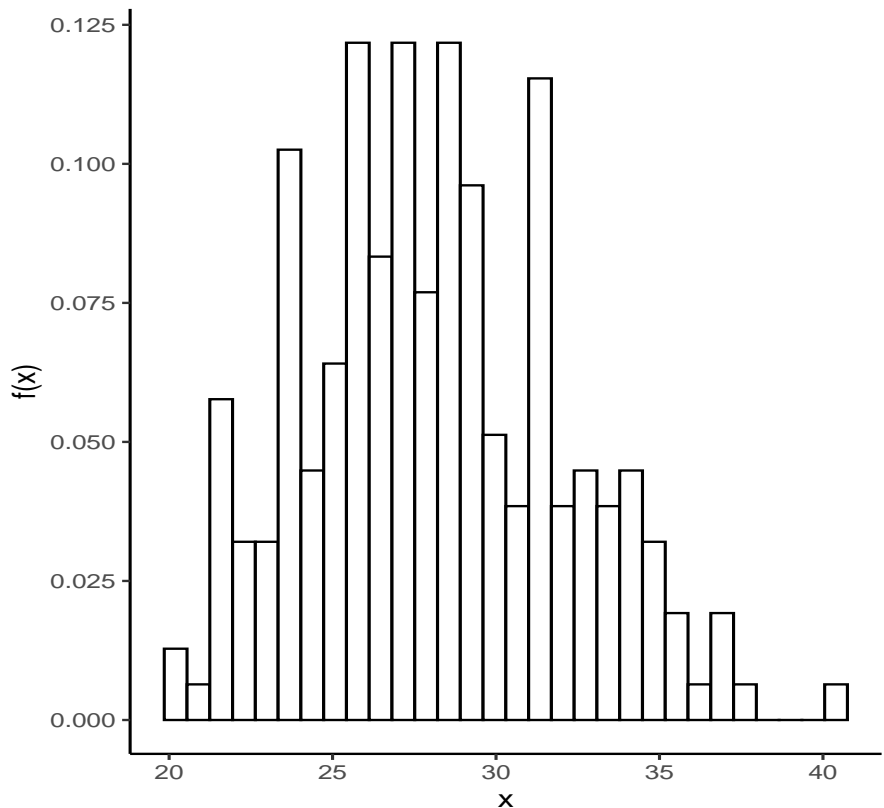}
	\caption{The distribution of weights of the sheep at seven months}
	\label{fig10}
\end{figure}
We note that this data set is available in \cite{hollander2013nonparametric}.

For the level-2 sampling design, we need $mk^{2}\leq N$. To obtain a ranked set sample, we follow the instructions in the Problem 32 \cite[Chapter~15]{hollander2013nonparametric}. We set $k=3$ and $m=7$ with reference to the Problem 32. First, the sheep are enumerated, $1\leq i\leq 224$. Then, three sheep are selected without replacement among the 224 sheep. By using the mother's weight at the time of mating as auxiliary variable, the selected sheep's weights are ranked from the smallest to the largest. After that, the sheep which has the smallest weight among the three sheep is selected and its weight is measured, $X_{1[1]1}=27.6$. The none of the sheep is returned to the other 221 sheep. After the first set, another three sheep are selected without replacement among the 221 sheep. The weights of the sheep are ranked from the smallest to the largest by using their mother's weight at the time of mating. In this set, the sheep which has the second smallest weight is selected and its weight is measured, $X_{2[2]1}=27.9$. The none of the sheep is returned to the other $218$ sheep. Then, three sheep are selected without replacement among the $218$ sheep. The first cycle is completed after the sheep which has the largest weight among the three sheep is selected and its weight is measured, $X_{3[3]1}=34.0$. The measured weights are given in the first row of Table \ref{Tab2}. The procedure is repeated in each cycle. Eventually, a ranked set sample of size $n=21$ is obtained by using level-2 sampling design. The weights of the $21$ sheep are given by Table \ref{Tab2}. First order inclusion probability of each sheep is $\pi_{i}^{(2)}=0.09375$ where $1\leq i\leq 224$. The second order inclusion probabilities are provided by Table S1 of the supplementary material. 
\begin{table}[!htp] 
	\small
	\centering
	\caption{A ranked set sample which is obtained by using level-2 sampling design}
	\begin{tabular}{l|lll}
		\hline
		& $1^{st}$ & $2^{nd}$ & $3^{rd}$ \\ \hline
		Cycle 1 & $27.6$ & $27.9$ & $34.0$ \\ 
		Cycle 2 & $25.5$ & $30.2$ & $25.5$ \\ 
		Cycle 3 & $26.5$ & $23.5$ & $25.9$ \\ 
		Cycle 4 & $23.0$ & $25.0$ & $40.5$ \\ 
		Cycle 5 & $20.5$ & $30.5$ & $35.1$ \\ 
		Cycle 6 & $23.0$ & $31.0$ & $33.5$ \\ 
		Cycle 7 & $27.9$ & $33.5$ & $35.5$ \\ \hline
	\end{tabular}
	\label{Tab2}
\end{table}

Now, we give an example for pointwise estimate of distribution function. It is assumed that the median $M$ is known to be $27.90$. Here, the goal is to estimate $F(27.90)$. By using the Eq. (7), we can write
\begin{equation*}
	\hat{F}_{L-2}^{\ast }(27.90)=\sum\limits_{i\in \pmb{D}_{2}}\frac{\frac{
			I\left( X_{i}\leq 27.90\right) }{0.09375}}{\sum\limits_{\imath\in \pmb{D}_{2}}\frac{1}{0.09375}}
\end{equation*}
where $1\leq i\leq 224$. It is found that $\hat{F}_{L-2}^{\ast }(27.90)=0.5714$. By using Eq. (11), $95\%$ confidence interval of $F(27.90)$ is obtained as $\left(0.3978, 0.745\right) $. On the other hand, we suppose that the median $M$ is not known. In this case, any quantile such as median can be estimated by using proposed EDF estimators. To find the estimator $\hat{M}$ of the median $M$, we set
\begin{equation*}
	\hat{M}=\hat{F}_{L-2}^{-1}(0.5)
\end{equation*}
where $\hat{F}_{L-2}^{-1}$ is inverse function of $\hat{F}_{L-2}^{\ast }$. First, the estimated distribution function $\hat{F}_{L-t}^{\ast }(x)$ is calculated by using Eq. (7). The probabilities $\hat{F}_{L-t}^{\ast }(x)$ are given by Table \ref{Tab3}.
\begin{table}[htt!] 
	\small
	\centering
	\caption{The values of $\hat{F}_{L-t}^{\ast }(x)$}
	\begin{tabular}{l|l}
		\hline
		$x$ & $\hat{F}_{L-t}^{\ast }(x)$ \\ \hline
		$x<20.5$ & $0$ \\ 
		$20.5\leq x<23.0$ & $0.0476$ \\ 
		$23.0\leq x<23.5$ & $0.1429$ \\ 
		$23.5\leq x<25.0$ & $0.1905$ \\ 
		$25.0\leq x<25.5$ & $0.2381$ \\ 
		$25.5\leq x<25.9$ & $0.3333$ \\ 
		$25.9\leq x<26.5$ & $0.3810$ \\ 
		$26.5\leq x<27.6$ & $0.4286$ \\ 
		$27.6\leq x<27.9$ & $0.4762$ \\ 
		$27.9\leq x<30.2$ & $0.5714$ \\ 
		$30.2\leq x<30.5$ & $0.6190$ \\ 
		$30.5\leq x<31.0$ & $0.6667$ \\ 
		$31.0\leq x<33.5$ & $0.7143$ \\ 
		$33.5\leq x<34.0$ & $0.8095$ \\ 
		$34.0\leq x<35.1$ & $0.8571$ \\ 
		$35.1\leq x<35.5$ & $0.9048$ \\
		$35.5\leq x<40.5$ & $0.9524$ \\ 
		$40.5\leq x$ & $1$ \\ \hline
	\end{tabular}
	\label{Tab3}
\end{table}
Then, the estimated median is obtained as $\hat{F}_{L-t}^{-1}(0.5)=\hat{M}=27.9$ according to Table \ref{Tab3}. Now, we define $c_{1}$ and $c_{2}$ to find approximate $95\%$ confidence interval for $M$. We set
\begin{equation*}
	\Pr \left\{ c_{1}\leq \hat{F}_{L-t}^{\ast }(M)\leq c_{2}\right\} =0.95
\end{equation*}
then, the $95\%$ confidence interval can be written as
\begin{equation*}
	\left(\hat{F}_{L-t}^{-1}(c_{1}),\hat{F}_{L-t}^{-1}(c_{2}) \right). 
\end{equation*}
By using Eq. (11), $c_{1}$ and $c_{2}$ can be expressed as
\begin{equation}
	c_{1}=0.5-1.96\left( \hat{V}\left[ \hat{F}_{L-t}^{\ast }(\hat{M})\right] \right)
	^{1/2}, \hspace{0.7cm} c_{2}=0.5+1.96\left( \hat{V}\left[ \hat{F}_{L-t}^{\ast }(\hat{M})\right] \right)
	^{1/2}. 
\end{equation}
By using Eqs. (9) and (14), $c_{1}$ and $c_{2}$ are calculated as follows:
\begin{equation*}
	\begin{split}
		c_{1}= & 0.5-1.96\sqrt{0.0072} \\ 
		=&0.3264 \\
		c_{2}= & 0.5+1.96\sqrt{0.0072} \\
		=&0.6736
	\end{split}
\end{equation*}
The $95\%$ confidence interval of $M$ is $\left(25.7112,30.7360 \right) $. Thus, we are $95\%$ confident that the median $M$ of the sheep's weights at seven months is between $25.71$ kg and $30.73$ kg.

\section{Conclusion}

In many studies such as environmental, ecological, agricultural, biological etc., researchers face with sampling problem. In general, sample observations are selected by using SRS without replacement. The protocol of SRS must be carefully planned since it is expected that each individual measurement in the sample is likely to be representative of the population characteristic of interest. Therefore, the $n$ observations should be selected from entire population. However, in practice, there is no guarantee that a single simple random sample of size $n$ is truly representative of the entire population. Of course, the problem can be solved simply if the sample size $n$ is increased by researcher. If taking actual measurements is difficult (e.g. costly and/or time consuming), increasing the sample size will not be a good solution. 

In this paper, we investigate three ranked-based sampling strategies which are called as level-0, level-1 and level-2. These sampling designs use additional information to create an artificially stratified population. In other words, $k$ artificial strata is deemed to consist when the set size is $k$. If the number of cycles is $m$, then $m$ measurements can be taken from each stratum and it makes possible to obtain more representative sample. On top of this statement, an interesting question has been pointed out by the anonymous referee and is as follows: Could the stratified simple random sampling (SSRS) with $k$-equal-sized strata be more efficient than RSS with set size $k$, since the strata wouldn't overlap as in RSS? In the procedure of RSS, obtaining $m$ units from each stratum is only an assumption. It is possible to obtain a different number of units from each stratum. However, it can be said that sampling designs provide more efficient estimator than the SSRS. Because, negative covariance will occur between any two units measured from different strata, especially when level-1 and level 2 sampling designs are used. It means that the negative covariance reduces the magnitude of the variances of the proposed estimators based on sampling designs. However, it cannot be observed a negative covariance between the any two measured units from different strata, when SSRS is used. Because, these samples which are selected from strata are totally independent.   

We have developed design-based estimators of distribution function for these sampling designs. We have examined the theoretical properties of the proposed estimators. Also, numerical results have been provided. Moreover, the EDF based on level-2 has been used to estimate the distribution function of sheep's weights at seven months. Thus, we can give some recommendations as follows:
\begin{enumerate}
	\item[1.] Level-2 sampling design shows outperformance among the three sampling designs since it is constructed by using without replacement policy. According to the REs of the proposed estimators w.r.t EDF based on SRS, we recommend to use the EDF based on level-2 sampling design.
	\item[2.] Figures \ref{fig1}-\ref{fig3} indicate that increase in the set size and the number of cycles enhances the REs between $0.2$ and $0.6$. Also, increase in the set size is showed to be more effective on REs. Depending on the difficulty of ranking the observations in the set, it is preferred to increase the set size rather than increase in the number of cycles.
	\item[3.] Regardless of the distribution of the interested population, the proposed estimator based on level-t sampling design ($t=0,1,2$) have been observed to be more efficient than the EDF based on the SRS even if the ranking is imperfect. However, the authors prefer that the correlation coefficient is greater than or equal to $0.75$. Because, the efficiency improvement diminishes as the ranking information becomes poor.
\end{enumerate}

In the literature, some authors such as Patil et al. \cite{patil1995finite}, Deshpande et al. \cite{deshpande2006nonparametric} Frey \cite{frey2011recursive}, Ozturk \cite{ozturk2014estimation} showed that Level-2 sampling scheme provides better statistical inference.   

In application, we provide an illustration for pointwise estimate of distribution function and estimation of the median by using EDF based on level-2 sampling design. We aim to show that $F(x)$ and $F^{-1}(p)$ can be estimated for a quantile $x$ and probability $p$, if the distribution function is estimated. Also, we give confidence intervals for $F(x)$ and $F^{-1}(p)$.

\section*{Appendix}
In this section, we provide a proof for the fourth part of Theorem 3.2. First, we show that $V\left[ \hat{F}_{h}(x)\right] =\frac{N-n}{N-1}\frac{\sigma ^{2}}{n}$ where $\sigma ^{2}=V\left[ I\left( X_{i}\leq x\right) \right] =F(x)\left(
1-F(x)\right) $, $i=1,\cdots, N$. Let us define another form of $V\left[ \hat{F}_{h}(x)\right]$ as follows:
\begin{equation*}
	V\left[ \hat{F}_{h}(x)\right] =\frac{\aleph +\Im }{N^{2}}.
\end{equation*}
Recall that $\pi _{i}=\frac{n}{N}$ and $\pi _{ii^{\prime }}=\frac{n(N-1)}{N(N-1)}$,
\begin{equation*}
	\begin{split}
		\aleph= & \sum\limits_{i=1}^{N}\left( \pi _{i}-\pi _{i}^{2}\right) \frac{%
			\left( I\left( X_{i}\leq x\right) -F(x)\right) ^{2}}{\pi _{i}^{2}} \\ 
		=&\left( \frac{n}{N}-\frac{n^{2}}{N^{2}}\right) \frac{N^{2}}{n^{2}}%
		\sum\limits_{i=1}^{N}\left( I\left( X_{i}\leq x\right) -F(x)\right) ^{2} \\
		=& N\left( \frac{N}{n}-1\right) \left( F(x)-F^{2}(x)\right), \\
		\Im= & \underset{i\neq i^{\prime }}{\sum \sum }\left( \pi _{ii^{\prime
		}}-\pi _{i}\pi _{i^{\prime }}\right) \left( \frac{I\left( X_{i}\leq x\right)
			-F(x)}{\pi _{i}}\right) \left( \frac{I\left( X_{i^{\prime }}\leq x\right)
			-F(x)}{\pi _{i^{\prime }}}\right) \\
		= & \left( \frac{N(n-1)}{n(N-1)}-1\right) \left(
		N^{2}F^{2}(x)-NF(x)-N(N-1)F^{2}(x)\right) \\
		= & -N\left( \frac{N(n-1)}{n(N-1)}-1\right) \left( F(x)-F^{2}(x)\right).
	\end{split}
\end{equation*}
Then,
\begin{equation}
	V\left[ \hat{F}_{h}(x)\right] =\frac{N-n}{N-1}\frac{\sigma ^{2}}{n}.
\end{equation}

Now, we define some notations to use in the rest of the proof. By using one of the designs which are level-0, level-1 and level-2, the following ranked set sample of size $mk$ is obtained.
$$%
\begin{array}{cccc}
	X_{(1)1} & X_{(2)1} & \cdots  & X_{(k)1} \\ 
	X_{(1)2} & X_{(2)2} & \cdots  & X_{(k)2} \\ 
	\vdots  & \vdots  &  & \vdots  \\ 
	X_{(1)m} & X_{(2)m} & \cdots  & X_{(k)m}%
\end{array}%
$$
In this matrix, $X_{(r)\tau}$ denotes the $r$th ranked unit from the $r$th set in $\tau$th cycle, $r=1,\cdots,k$ and $\tau=1,\cdots,m$. Also, we assume that $C(r,s)=Cov\left[ I\left( X_{(r)\tau }\leq x\right) ,I\left( X_{(s)\tau
	^{\prime }}\leq x\right) \right] $. By Patil et al. \cite{patil1995finite}, $C(r,s)$ were described as following.
\begin{equation*}
	C(r,s)=\pmb{\Psi} ^{T }\left( \pmb{\Delta} _{rs}-\pmb{\nabla} _{r}\pmb{\nabla}
	_{s}^{T }\right) \pmb{\Psi}
\end{equation*}
where $\pmb{\Psi}=\left( I\left( X_{1}\leq x\right) ,\cdots ,I\left( X_{N}\leq x\right)
\right) ^{T }$, $\pmb{\nabla} _{r}$ is $N$ dimensional column vector and $\pmb{\Delta} _{rs}$ is $N\times N$ dimensional matrix, $1\leq r,s\leq k$. We note that $\pmb{\Psi} ^{T }$ is the transpose of the vector $\pmb{\Psi}$. Here, $\pmb{\nabla} _{r}$ includes the components $\nabla _{r}^{i}$ which is probability that $r$th ranked unit in the set is $i$th ranked unit in the population. Also, $\pmb{\Delta} _{rs}$ includes the components $\Delta _{rs}^{ii^{\prime }}$ which is probability that $r$th ranked unit from a set has rank $i$ in the population and $s$th ranked unit from another set has rank $i^{\prime }$ in the population, $1\leq i,i^{\prime }\leq N$. Thus, it is clearly seen that $C(r,s)$ vary depending RSS design.

As in Patil et al. \cite{patil1995finite}, we define a variance form for $\hat{F}_{L-t}^{\ast }(x)$ as following.
\begin{equation}
	\left( mk\right) ^{2}V\left[ \hat{F}_{L-t}^{\ast }(x)\right]
	=m\sum\limits_{r=1}^{k}\sigma
	_{(r)}^{2}+m^{2}\sum\limits_{r=1}^{k}\sum\limits_{s=1}^{k}C(r,s)-m\sum%
	\limits_{r=1}^{k}C(r,r),
\end{equation}
where $\sigma
_{(r)}^{2}=V\left[I\left( X_{(r) }\leq x\right) \right] $. To define the first summation of the right hand side in (16), it is assumed that a simple random of size $n=mk$ partition into $m$ subsamples of size $k$. Considering that each subsample is ranked from the smallest to the largest, a form of $V\left[ \hat{F}_{h}(x)\right]$ can be obtained as follows:
\begin{equation}
	V\left[ \hat{F}_{h}(x)\right] =\frac{1}{mk^{2}}\sum\limits_{\tau =1}^{m}%
	\left[ \sum\limits_{r=1}^{k}\sigma _{(r)}^{2}+\underset{r\neq s}{\sum
		\sum }Cov\left( I\left( X_{(r)\tau }\leq x\right) ,I\left( X_{(s)\tau }\leq
	x\right) \right) \right]
\end{equation}
By using the Eqs. (15) and (17), the following equation is obtained.
\begin{equation}
	\sum\limits_{r=1}^{k}\sigma _{(r)}^{2}=k\frac{N-mk}{N-1}\sigma
	^{2}-\sum\limits_{\tau =1}^{m}\underset{r\neq s}{\sum \sum }Cov\left(
	I\left( X_{(r)\tau }\leq x\right) ,I\left( X_{(s)\tau }\leq x\right) \right),
\end{equation}
where $\sigma ^{2}=\frac{mk\left( N-1\right) }{N-mk}V\left[ \hat{F}_{h}(x)%
\right] $.
Finally, $V\left[ \hat{F}_{L-t}^{\ast }(x)\right]$ is obtained as follows:
\begin{align*}
	V\left[ \hat{F}_{L-t}^{\ast }(x)\right] & =V\left[ \hat{F}_{h}(x)\right] -%
	\frac{1}{mk^{2}} \left( \sum\limits_{\tau =1}^{m}\underset{r\neq s}{\sum
		\sum }Cov\left( I\left( X_{(r)\tau }\leq x\right) ,I\left( X_{(s)\tau }\leq
	x\right) \right) \right.\nonumber\\
	&\qquad \left. {} \hspace{3.2cm} -m\sum\limits_{r=1}^{k}\sum\limits_{s=1}^{k}C(r,s)+\sum%
	\limits_{r=1}^{k}C(r,r) \right)
\end{align*}
According to the theoretical results in Takahasi and Futatsuya \cite{takahasi1998dependence}, we can say that $Cov\left( I\left( X_{(r)\tau }\leq x\right) ,I\left( X_{(s)\tau }\leq
x\right) \right)\geq 0$, $C(r,s)\leq 0$ and $C(r,r)\leq 0$. Considering that $m\sum\limits_{r=1}^{k}\sum\limits_{s=1}^{k}C(r,s)>\sum%
\limits_{r=1}^{k}C(r,r)$, the proof is complete.

\bibliographystyle{tfnlm}
\bibliography{interactnlmsample}
	
\end{document}